\documentclass[11pt, letterpaper]{article}
\title{Improved Monotonicity Testers via Hypercube Embeddings}
\author{
Mark Braverman
\thanks{Department of Computer Science, Princeton University.
Research supported in part by the NSF Alan T. Waterman Award, Grant No. 1933331, a Packard Fellowship in Science and Engineering, and the Simons Collaboration on Algorithms and Geometry. }\and
Subhash Khot\thanks{Courant institute of Mathematical Sciences, New York University. Supported by
		the NSF Award CCF-1422159, NSF Award CCF-2130816, and the Simons Investigator Award.} \and
Guy Kindler\thanks{Engineering and Computer Science Department, the Hebrew University. Supported by Israel Science Foundation grant no.\ 2635/19.} \and
Dor Minzer\thanks{Department of Mathematics, Massachusetts Institute of Technology.  Supported by a Sloan Research Fellowship.}
}
\date{\vspace{-5ex}}
\usepackage{fullpage}
\usepackage{amsthm}
\usepackage{amsmath,amssymb,amsfonts,nicefrac}
\usepackage{xspace}
\usepackage{color}
\usepackage{url}
\usepackage{hyperref}
\usepackage{bm}
\usepackage{bbm}
\usepackage{times}

\usepackage{enumitem}

\newtheorem{thm}{Theorem}[section]

\newtheorem{lemma}[thm]{Lemma}
\newtheorem{corollary}[thm]{Corollary}
\newtheorem{claim}[thm]{Claim}

\newtheorem{definition}[thm]{Definition}

\newcommand\card[1]{\left| {#1} \right|}
\newcommand\sett[2]{\left\{ \left. #1 \;\right\vert #2 \right\}}

\newcommand\Prob[2]{{\Pr_{#1}\left[ {#2} \right]}}

\newcommand\cProb[3]{{\Pr_{#1}\left[ \left. #3 \;\right\vert #2 \right]}}

\newcommand\Expect[2]{{\mathop{\mathbb{E}}_{#1}\left[ {#2} \right]}}

\newcommand\eps{\varepsilon}

\renewcommand\geq{\geqslant}
\renewcommand\leq{\leqslant}

\newtheorem*{rep@theorem}{\rep@title}
\newcommand{\newreptheorem}[2]{%
\newenvironment{rep#1}[1]{%
\def\rep@title{\bf #2 \ref*{##1} \text{(Restated)} }%
\begin{rep@theorem} }%
{\end{rep@theorem} } }

\newtheorem*{rep@claim}{\rep@title}
\newcommand{\newrepclaim}[2]{%
\newenvironment{rep#1}[1]{%
\def\rep@title{\bf #2 \ref*{##1} \text{(Restated)} }%
\begin{rep@claim} }%
{\end{rep@claim} } }

\newreptheorem{theorem}{Theorem}
\newrepclaim{claim}{Claim}

\newcommand{\changed}[1]{}

\begin{document}
\maketitle
\begin{abstract}
  We show improved monotonicity testers for the Boolean hypercube under the $p$-biased measure, as well as over the
  hypergrid $[m]^n$. Our results are:
  \begin{enumerate}
    \item For any $p\in (0,1)$, for the $p$-biased hypercube we show a non-adaptive tester that makes
    $\tilde{O}(\sqrt{n}/\eps^2)$ queries, accepts monotone functions with probability $1$ and
    rejects functions that are $\eps$-far from monotone with probability at least $2/3$.
    \item For all $m\in\mathbb{N}$, we show an $\tilde{O}(\sqrt{n}m^3/\eps^2)$ query monotonicity tester over $[m]^n$.
  \end{enumerate}
  We also establish corresponding directed isoperimetric inequalities in these domains, analogous to the isoperimetric inequality in~\cite{KMS18}.
  Previously, the best known tester due to Black, Chakrabarty and Seshadhri~\cite{BCS18} had $\Omega(n^{5/6})$ query complexity.
  Our results are optimal up to poly-logarithmic factors and the dependency on $m$.

  Our proof uses a notion of monotone embeddings of measures into the Boolean hypercube that can be used to
  reduce the problem of monotonicity testing over an arbitrary product domains to the Boolean cube.
  The embedding maps a function over a product domain of dimension $n$ into a function over a Boolean cube of a larger dimension $n'$, while preserving its distance from being monotone; an embedding is considered efficient if $n'$ is not much larger than $n$, and we show how to construct efficient embeddings in the above mentioned settings.
\end{abstract}
\section{Introduction}
Let $[m] = \{0,1,\ldots,m-1\}$ be thought of as an ordered set, and consider the partial ordering induced by it on $[m]^n$:
two points $x,y\in [m]^n$ satisfy $x\leq y$ if and only if $x_i\leq y_i$ for all $i=1,\ldots,n$. A function $f\colon [m]^n\to\{0,1\}$ is called {\em monotone} if
for every $x,y\in[m]^n$ such that $x\leq y$ we have $f(x)\leq f(y)$. Given a function $f\colon [m]^n\to\{0,1\}$, we measure its distance from being monotone, with respect to a probability measure $\mu$ over $[m]^n$, by
\[
\eps(f;\mu) = \min_{g\colon [m]^n\to\{0,1\}\text{ monotone}} \Delta(f,g; \mu),
\qquad
\text{where }
\Delta(f,g; \mu) = \mu\left(\sett{x\in[m]^n}{f(x)\neq g(x)}\right).
\]
In this paper we present monotonicity testers for functions over $[m]^n$ (under the uniform measure), as well as over the Boolean hypercube $\{0,1\}^n$ with
the $p$-biased measure, defined as $\mu_p^{\otimes n}(x) = p^{\card{x}}(1-p)^{n-\card{x}}$. That is, we construct a randomized algorithm that
makes oracle queries to an unknown function $f$ over the domain, which accepts with probability $1$ if $f$ is monotone, and rejects with
probability $\geq\frac{2}{3}$ if $f$ is $\eps$-far from monotone.

\subsection{Prior Works}
The monotonicity testing problem has received significant attention over the years, as we shall now review. For simplicity, below we think of $\eps$ as being a small constant.
The problem was originally studied over Boolean hypercube with the uniform measure~\cite{GGLRS},
where a non-adaptive algorithm that makes $O(n)$ queries was shown. This bound was improved by~\cite{CS16}, who showed an $\tilde{O}(n^{7/8})$ query tester by proving
a directed version of Margulis' isoperimetric inequality~\cite{Margulis} and using it towards developing improved monotonicity testers.
Following it, Chen, Servedio and Tan~\cite{CST14} modified the algorithm and the analysis of~\cite{CS16} and established a $\tilde{O}(n^{5/6})$ query tester. Finally, an $\tilde{O}(\sqrt{n})$ query tester was given in~\cite{KMS18}, who proved a directed version of
an isoperimetric inequality due to Talagrand~\cite{Talagrand}. The tester of~\cite{KMS18} is tight up to poly-logarithmic factors, by a bound for non-adaptive testers due to~\cite{FLNRRS}. The
best known lower bound for adaptive testers is not too far off~\cite{CWX17}, and currently stands at $\tilde{\Omega}(n^{1/3})$.

Following the investigation of complexity of monotonicity testing over the hypercube, variants of the problem were also considered in the literature, in which
either the domain or the range of the function are different~\cite{dodis1999improved,CS14,BCS18,BCS20,BKR}. Most relevant to us is the monotonicity testing
problem over different measures on the Boolean hypercube as well as the closely related hypergrid $[m]^n$, wherein the state of the art result is an $O(n^{5/6}{\sf poly}(\log m))$
query tester due to~\cite{BCS18}. To prove their result, the authors of~\cite{BCS18} established an analog of the directed isoperimetric
inequality of~\cite{CS16} for the hypergrid.

\subsection{Parallel Works}
Following initial submission of this paper, we have learned that Black, Chakrabarty and Seshadhri have independently obtained results similar to ours~\cite{BCSnew}.
They use a different method, first proving analogous directed isoperimetric inequalities over the hypergrid, and then using these to construct and analyze a monotonicity
tester for the hypergrid.

\subsection{Main Results}
Our first result is an essentially-optimal monotonicity tester for the $p$-biased cube:
\begin{thm}\label{thm:1}
  For every $p\in (0,1)$, there is a non-adaptive monotonicity tester over $(\{0,1\}^n,\mu_p^{\otimes n})$ that makes $\tilde{O}(\sqrt{n}/\eps^2)$
  queries.
\end{thm}

Second, we focus on the hypergrid $[m]^n$. Here and throughout, $U_m$ refers to the uniform distribution over $[m]$, and we often drop the subscript $m$ when it is clear from
the context.
\begin{thm}\label{thm:2}
  For all $m,n\in\mathbb{N}$, there is a non-adaptive monotonicity tester over $([m]^n,U^{\otimes n})$ that makes $\tilde{O}(\sqrt{n}m^3/\eps^{2})$ queries.
\end{thm}

Our techniques also imply analogs of the directed isoperimetric inequality of~\cite{KMS18} for the hypergrid as well as for the
$p$-biased cube. For simplicity we state the result for the hypergrid, and defer the statement for the $p$-biased cube to Theorem~\ref{thm:directed_isoperimetric_biased}.

Let $f\colon [m]^n\to \{0,1\}$ be a function and fix an input $x\in [m]^n$.
The negative sensitivity of $f$ at $x$, denoted by $s_f^{-}(x)$, is defined to be the number of coordinates $i$ such that there is an input $y$ differing from $x$ only on the $i^{\text{th}}$ coordinate, such that the pair $(x,y)$ violates monotonicity. Namely, it is the number of coordinates $i$ such that for the point $y$ which differs from $x$ only on its
$i^{\text{th}}$ coordinate, we have that $x<y$ and $f(x)>f(y)$ (if $f(x) = 1$) or $x>y$ and $f(x)<f(y)$ (if $f(x) = 0$).
\begin{thm}\label{thm:directed_isoperimetric}
  If $f\colon[m]^n\to\{0,1\}$ is $\eps$-far from monotone with respect to $U^{\otimes n}$, then
  \[
      \Expect{x\in [m]^n}{\sqrt{s_f^{-}(x)}}\geq \Omega\left(\frac{\eps}{m^3\log(mn/\eps)^2}\right).
  \]
\end{thm}

\subsection{Our Technique}
Our proofs rely on the following elementary notion of an embedding of a domain (that we wish to test monotonicity over) into
a hypercube of not too-large dimension.
\begin{definition}\label{def:embed}
  We say that a probability distribution $([m], \mu_1)$ can be $r$-locally embedded if there is a Boolean hypercube
  $\{0,1\}^r$, a map $\phi\colon \{0,1\}^r\to [m]$, a collection $\Psi = \{\Psi_{\omega}\}_{\omega\in \Omega}$ of
maps $\Psi_{\omega}\colon [m]\to\{0,1\}^r$, and a probability distribution $P$ over $\Omega$ such that:
  \begin{enumerate}
  \item Each one of $\phi$ and $\Psi_{\omega}$ is monotone.
    \item Sampling $x\sim U(\{0,1\}^r)$, the distribution of $\phi(x)$ is $\mu_1$.
    \item \label{simsim} Sampling $y\sim \mu_1$ and $\omega\sim P$, the distribution of $\Psi_{\omega}(y)$ is uniform over $\{0,1\}^r$.
    \item For each $\omega\in \Omega$, the composition $\phi\circ \Psi_{\omega}$ is
        the identity on $[m]$.
  \end{enumerate}
\end{definition}

The usefulness of Definition~\ref{def:embed} comes from the fact that
given a local embedding of $[m]$, we can reduce the problem of testing
monotonicity over $([m]^{n},\mu_1^{\otimes n})$ to that of testing
it over Boolean hypercubes of dimension $rn$, which we already know
how to solve.
Towards showing the reduction we note that if $([m],\mu_1)$ can be
$r$-locally embedded, then given
a function $f\colon [m]^n\to \{0,1\}$ we may define $g\colon \{0,1\}^{r\times n}\to \{0,1\}$ by
\[
g(x(1),\ldots,x(n)) = f(\phi(x(1)),\ldots,\phi(x(n))).
\]
The following lemma asserts that if $f$ is monotone then $g$ is also monotone, and if $f$ is $\eps$-far from monotone, then $g$ is $\eps$-far from monotone.
\begin{lemma}\label{lem:pres_dist}
  If $f$ is monotone, then $g$ is monotone. Moreover, $\eps(g; U^{\otimes n}) \geq \eps(f; \mu_1^{\otimes n})$.
\end{lemma}
\begin{proof}
    Assume $f$ is monotone. Then taking any
    $(x(1),\ldots, x(n))\leq ({x(1)}',\ldots,{x(n)}')$ in
    $\{0,1\}^{r\times n}$ we have by the monotonicity of $\phi$ that
    $(\phi(x(1)),\ldots,\phi(x(n)))\leq
    (\phi({x(1)}'),\ldots,\phi({x(n)}'))$, and using the monotonicity
    of $f$ we get that
    $g(x(1),\ldots,x(n))\leq g({x(1)}',\ldots,{x(n)}')$.

  For the other direction, let $g'$ be the closest monotone function to $g$, and choose $\vec{\omega}=(\omega_1,\ldots,\omega_n)\sim P^{\otimes n}$. Define
  \[
  f_{\vec{\omega}}(x_1,\ldots,x_n) = g'(\Psi_{\omega_1}(x_1),\ldots,\Psi_{\omega_n}(x_n)).
  \]
  Since each $\Psi_{\omega_i}$ is monotone we have that
  $f_{\vec{\omega}}$ is monotone as well. Also,
  \begin{align*}
  \Expect{\vec{\omega}}{\Delta(f, f_{\vec{\omega}};\mu_1^{\otimes n})}
  =\Expect{\vec{\omega}}{\Expect{x\sim \mu_1^{\otimes n}}{1_{f(x)\neq f_{\vec{\omega}}(x)}}}
  &=\Expect{\vec{\omega}}{\Expect{x\sim \mu_1^{\otimes n}}{1_{f(x)\neq g'(\Psi_{\omega_1}(x_1),\ldots,\Psi_{\omega_n}(x_n))}}}\\
  &=\Expect{\vec{\omega}}{\Expect{x\sim \mu_1^{\otimes n}}{1_{g(\Psi_{\omega_1}(x_1),\ldots,\Psi_{\omega_n}(x_n))\neq g'(\Psi_{\omega_1}(x_1),\ldots,\Psi_{\omega_n}(x_n))}}},
  \end{align*}
  where in the last equality we used the fact that $\phi\circ
  \Psi_{\omega} $ is the identity. Note that  by property \ref{simsim}
  of an embedding, given the distribution of $\vec{\omega}$ and $x$,
  the distribution of
  $(\Psi_{\omega_1}(x_1),\ldots,\Psi_{\omega_n}(x_n))$ is uniform over
  $\{0,1\}^{r\times n}$, so the last expression is equal to
  $\Delta(g,g'; U^{\otimes n}) = \eps(g; U^{\otimes n})$. It follows that there is an $\vec{\omega}$ such that $\Delta(f, f_{\vec{\omega}};\mu_1^{\otimes n})\leq \eps(g; U)$,
  and the proof is concluded.
\end{proof}

For the Boolean hypercube with the uniform measure, a $2$-query path
tester is constructed in~\cite{KMS18} which always accepts monotone
functions, and rejects functions that are $\eps$-far from monotone
with probability at least
\begin{equation}\label{eq:R}
R(n,\eps)= \frac{\eps^2}{\sqrt{n}{\sf poly}(\log n)}.
\end{equation}
Combining that
tester
with Lemma~\ref{lem:pres_dist} we get the following conclusion:
\begin{lemma}\label{lem:tester_conc}
  Suppose that $([m],\mu_1)$ can be $r$-locally embedded; then there is a $2$-query monotonicity testing algorithm for functions over $([m]^{n},\mu_1^{\otimes n})$
  that always accepts monotone functions, and rejects functions that
  are $\eps$-far from monotone with probability at least $R(rn,\eps)$.
\end{lemma}
\begin{proof}
  Given $f\colon ([m]^n,\mu_1^{\otimes n})\to\{0,1\}$, define $g$ as
  above, then run the monotonicity tester of the hypercube on $g$, and accept/reject accordingly.
  Note that a single query to $g$ can be answered by making a single query to $f$. By Lemma~\ref{lem:pres_dist}, if $f$ is monotone then $g$ is monotone, hence the
  tester always accepts. If $f$ is $\eps$-far from monotone, then by Lemma~\ref{lem:pres_dist} $g$ is also $\eps$-far from monotone, hence the tester rejects
  with probability at least $R(rn,\eps)$.
\end{proof}

Thus, Theorems~\ref{thm:1} and~\ref{thm:2} follow from Lemma~\ref{lem:tester_conc} once we show the existence of sufficiently good local embeddings.
In Section~\ref{sec:consturct} we show constructions of such embeddings for the $p$-biased measure on $\{0,1\}$, as well as basic embeddings for $[m]^n$ which
are not good enough for our purpose (but gives some intuition). To construct efficient embeddings for $[m]^n$ we have to work harder, and for divisibility reasons
we only know how to construct such embeddings for $m$'s that are power of $2$. For other $m$'s, we have to consider a slightly relaxed notion of embeddings,
asserting that there are distributions $\mu_1'$ and $\mu_2'$ that are extremely close to the distributions $([m], \mu_1)$ and $(\{0,1\}^r, U)$ such that one can embed
$([m], \mu_1')$ into $(\{0,1\}^r, \mu_2')$; see Sections~\ref{sec:match_def},~\ref{sec:almost_PM} for the formal definition. This relaxed notion has the
same monotonicity testing and directed isoperimetric implications.
The construction of embeddings for the hypergrid is more involved than our construction of embeddings for the $p$-biased cube, and can be found in Section~\ref{sec:mono_match_construct}.

As for the directed isoperimetric inequalities, we recall the isoperimetric result from~\cite{KMS18}
\begin{thm}\label{thm:directed_tal_cube}
  If $f\colon(\{0,1\}^n,U^{\otimes n})\to\{0,1\}$
is $\eps$-far from monotone, then $\Expect{x}{\sqrt{s_f^{-}(x)}}\geq
\Omega\left(\frac{\eps}{\log(n/\eps)}\right)$.
\end{thm}
Combining Theorem~\ref{thm:directed_tal_cube} with Lemma~\ref{lem:pres_dist}
we get:
\begin{lemma}\label{lem:embed_to_iso}
  Suppose that $([m],\mu_1)$ can be $r$-locally embedded; then for any $f\colon ([m]^n,\mu_1)\to\{0,1\}$ that is $\eps$-far from monotone it holds that
  $\Expect{y\sim \mu_1^{\otimes n}}{\sqrt{s_f^{-}(y)}}\geq \Omega\left(\frac{\eps(f)}{\sqrt{r}\log(rn/\eps(f))}\right)$.
\end{lemma}
\begin{proof}
  Define $g\colon \{0,1\}^{r\times n}\to\{0,1\}$ as above, and note that $\frac{1}{r}s_g^{-}(x)\leq s_f^{-}(\phi(x))$ for all $x\in\{0,1\}^{r\times n}$.
  Indeed, letting $k=\frac{1}{r}s_g^{-}(x)$ and viewing $x = (x(1),\ldots,x(n))$ where $x(i)\in \{0,1\}^r$, there are at least $k$ $i$'s such that
  there is $x'$ such that $x'(j) = x(j)$ for all $j\neq i$ and the pair $x,x'$ violates monotonicity of $g$. In that case, we see that the pair
  $y=\phi(x) = (\phi(x(1)),\ldots,\phi(x(n)))$ and $y'=\phi(x')=(\phi(x'(1)),\ldots,\phi(x'(n)))$ only differ in their $i^{\text{th}}$ coordinate and violate monotonicity of $f$,
  hence $s_f^{-}(\phi(x))\geq k$. It follows that
  \[
  \Expect{y\sim \mu_1^{\otimes n}}{\sqrt{s_f^{-}(y)}}
  =
  \Expect{x\in \{0,1\}^{r\times n}}{\sqrt{s_f^{-}(\phi(x))}}
  \geq
  \frac{1}{\sqrt{r}}\Expect{x\in \{0,1\}^{r\times n}}{\sqrt{s_g^{-}(x)}}
  \geq
  \frac{1}{\sqrt{r}}\Omega\left(\frac{\eps(g)}{\log(rn/\eps(g))}\right),
  \]
  and the proof is concluded by Lemma~\ref{lem:pres_dist}.
\end{proof}

We note that Theorem~\ref{thm:directed_isoperimetric} follows from Lemma~\ref{lem:embed_to_iso} (or rather, a slight variant of it which is suitable for
slightly relaxed embeddings) by showing that $([m],U)$ can be $r$-locally embedded for $r=O(m^6)$
(under the aforementioned slightly relaxed notion of embeddings).

\section{Elementary Constructions of Embeddings}\label{sec:consturct}
In this section we present several ideas for constructing local embeddings and prove Theorem~\ref{thm:1}.

\subsection{Embedding $p$-biased Cubes}
We begin by constructing some basic embeddings for $p$-biased distributions over $\{0,1\}$, and then combining them to prove Theorem~\ref{thm:1}.
First, we show that the measure $\mu_{p}$ can be locally embedded when
$p$ is a powers of $2$.
\begin{lemma}\label{lem:embed_p_1}
  Let $p = 2^{-r}$, and consider the distribution $\mu_p$ over $\{0,1\}$ where $\mu_p(1) = p$. Then $(\{0,1\},\mu_p)$ can be $r$-locally embedded.
\end{lemma}
\begin{proof}
  We define $\phi(x_1,\ldots,x_r) = x_1\land\ldots\land x_r$. As for
  $\Psi$, we take the distribution $(\Omega,P)$ to be uniform over $\{0,1\}^{r}\setminus\{\vec{1}\}$,
  and define $\Psi_{\omega}(1) = (1,\ldots,1)$ and $\Psi_{\omega}(0) = \omega$.
\end{proof}

\noindent Secondly, we show that if $\mu_{p}$ can be locally embedded, then so can $\mu_{1-p}$.
\begin{lemma}\label{lem:embed_p_2}
  Let $p \in (0,1)$, and suppose $(\{0,1\},\mu_p)$ can be $r$-locally embedded. Then $(\{0,1\},\mu_{1-p})$ can be $r$-locally embedded.
\end{lemma}
\begin{proof}
  Let $(\phi,\{\Psi_{\omega}\}_{\omega\in \Omega},P)$ be an $r$-local embedding of $\mu_p$. Define
  $\phi'(x) = 1-\phi(1-x)$ and $\Psi'_{\omega}(a) = \vec{1}-\Psi_{\omega}(1-a)$. First, note that
  $\phi'$ and $\Psi'_{\omega}$ are monotone. Second, sampling $x\sim U$, $\phi'(x)$ is distributed according to $\mu_{1-p}$. Also,
  \[
  \phi'(\Psi'_{\omega}(a))
  =1-\phi(1-\Psi'_{\omega}(a))
  =1-\phi(\Psi_{\omega}(1-a))
  =1-(1-a) = a.
  \]
  Finally, if $a\sim \mu_{1-p}$, then $1-a\sim \mu_p$, hence $\Psi_{\omega}(1-a)\sim U$ and so $\Psi'_{\omega}(a)\sim U$.
\end{proof}

Third, we show how $\mu_{p_1p_2}$ can
be locally embedded given local embeddings for $\mu_{p_1}$ and $\mu_{p_2}$.
\begin{lemma}\label{lem:embed_p_3}
  Suppose that $\mu_{p_1}$ can be $r_1$ locally embedded, and $\mu_{p_2}$ can be $r_2$ locally embedded. Then
  $\mu_{p_1p_2}$ can be $r_1+r_2$ locally embedded.
\end{lemma}
\begin{proof}
    Let $(\phi_1, \{\Psi_{1,\omega}\}_{\omega\in \Omega_1}, P_1)$ and
    $(\phi_2, \{\Psi_{2,\omega}\}_{\omega\in \Omega_2}, P_2)$ be the
    local embeddings of $\mu_{p_1}$ and $\mu_{p_2}$, respectively. We
    define
    $\phi\colon \{0,1\}^{r_1+r_2}\to \{0,1\}$ by $\phi(x,y) =
    \phi_1(x)\land \phi_2(y)$.

Now let $\Omega' = \{0,1\}^2\setminus \{(1,1)\}$, and define $P'$ to be the
distribution obtained by taking  $(a,b)\sim\mu_{p_1}\times
\mu_{p_2}$ and conditioning on the event $[(a,b)\neq (1,1)]$.
We take $\Omega = \Omega_1\times \Omega_2\times \Omega'$ and $P =
P_1\times P_2\times P'$. For $w=(\omega_1,\omega_2,\omega')\in
\Omega$ where $w'=(a,b)$, we finally define $\Psi_{\omega}$ as follows:
  \[
  \Psi_{(\omega_1,\omega_2,\omega')}(1) = (\Psi_{1,\omega_1}(1), \Psi_{2,\omega_2}(1)),
\]
and
  \[
  \Psi_{(\omega_1,\omega_2,\omega')}(0) = (\Psi_{1,\omega_1}(a), \Psi_{2,\omega_2}(b)).
  \]
  It is clear that $\phi$ and $\Psi$ are monotone, that
  $\phi\circ \Psi_{\vec{\omega}} = {\sf identity}$, and that the
  distributions are correct.
\end{proof}

Next, by an easy approximation argument we conclude that for all
values of $p$ there is some $p'$ close to $p$ such that $\mu_{p'}$ can be locally embedded.
\begin{corollary}\label{cor:approx_arg}
  For all $\delta>0$ and $p\in (0,1)$, there exists a  $p'\in (0,1)$
  such that $\card{p-p'}\leq \frac{\delta}{n^{10}}$ and that
  $\mu_{p'}$ can be $O(\log^2(n/\delta))$-locally embedded.
\end{corollary}
\begin{proof}
    By Lemma~\ref{lem:embed_p_2}, we may assume that $p\leq 1/2$. Set
    $s = \lceil 10\log(n/\delta)\rceil$, and for $a\in \mathbb{Z}^s$, define
    $q(\vec{a})=\prod\limits_{i=1}^{s}(1-2^{-i})^{a_i}$.

    Below, we  show that there exists a vector $a\in \mathbb{Z}^s$ such that
    $p'=q(a)$ satisfies
    \begin{equation}
        p\leq p'\leq p+\frac{\delta}{n^{10}},\label{eq:pprime}
    \end{equation}
    and where $a_1\in\{1,\ldots, s\}$ and $a_i\in \{0,1,2,3\}$ for
    any $i>1$. Note that this implies,  by
    Lemmas~\ref{lem:embed_p_1},~\ref{lem:embed_p_2}, and~\ref{lem:embed_p_3},
    that $p'$ can be $r$-locally embedded for
    $r\leq s+O(s^2) = O(\log^2(n/\delta))$, finishing the proof.

    \medskip
    To find the required vector $a$, we begin by taking $k$ to be
    the maximal number that satisfies $\frac1{2^k}\geq p$. If
    $k\geq s$ we set $a_1=s$, and note that we are done since the
    vector $a=(a_1,0,\ldots, 0)$
    satisfies  \eqref{eq:pprime} as required. Otherwise if $k<s$, we
    continue to
    set $a_1=k$, and define $a^1=(a_1,0,\ldots, 0)$. We then
    go over $i=2,\ldots, s$, finding at each step the
    largest number $k$ that satisfies $q(a^{i-1}+k\cdot e_i)\geq p$,
    and then taking $a_i=k$ and  $a^i=(a_1,\ldots,a_i, 0, \ldots, 0)$
    (here $e_i$ is the $i^{th}$ unit vector).

    We set our final vector to be $a=a^s$. It follows immediately from
    the definition of $a^s$ that $q(a)\geq p$ and that
    $q(a)\cdot (1-2^{-s})<p$, which implies that
    $q(a)\leq p \cdot (1-2^{-s})^{-1}\leq  p\cdot(1+2\cdot 2^{-s})\leq
    p+2^{-s}\leq p+\frac{\delta}{n^{10}}$.  We therefore have that $a$
    satisfies \eqref{eq:pprime}. It is also clear from the definition
    that $a_1\in\{1,\ldots, s\}$. To show that $a_i\in \{0,1,2,3\}$
    for all $i>1$, we first observe that it is clear from the definition
    of the $a_i$'s that for all $i$, $\ p\leq q(a^i) \leq p\cdot (1-2^{-i})^{-1}$. It then follows for each $i>1$ that
    \begin{equation*}
        \label{eq:1}
    q(a^{i-1}+4\cdot e_i)  \leq  p\cdot  (1-2^{-(i-1)})^{-1}\cdot(1-2^{-i})^4<p,
\end{equation*}
as can be verified by a simple application of the binomial expansion to $(1-2^{-i})^4$. The definition of $a_i$
therefore dictates that $a_i<4$, as desired.
\end{proof}

\paragraph{Proof of Theorem~\ref{thm:1}.} Notice that if
$|p'-p|\leq \frac{\eps}{2n}$, a function
$f:\{0,1\}^n\to\{0,1\}$ which is $\eps$ far from monotone with
respect to $\mu_p^n$ is $\frac{\eps}{2}$ far from monotone with
respect to the measure $\mu_{p'}^n$. Hence it is enough to apply a
monotonicity testing algorithm to $f$ with respect to $\mu_{p'}^n$. We
thus use Corollary~\ref{cor:approx_arg} to find a $p'$ that is
sufficiently close to $p$ and that is $r$-locally embeddable for
$r=O(\log^2(n/\eps))$, and then apply the tester from
Lemma~\ref{lem:tester_conc} with respect to the measure $\mu_{p'}^n$
and the error $\eps/2$. To obtain Theorem~\ref{thm:1}, we
independently repeat this tester $\frac
{10}{R(rn,\eps/2)} = \tilde{O}(\sqrt{n}/\eps^2)$ times.

\paragraph{A directed isoperimetric inequality over the $p$-biased hypercube.} By
Corollary~\ref{cor:approx_arg} and Lemma~\ref{lem:embed_to_iso}, we
get an analog of Theorem~\ref{thm:directed_tal_cube} for the
$p$-biased cube, stated below.
\begin{thm}\label{thm:directed_isoperimetric_biased}
  For all $p\in (0,1)$, if $f\colon(\{0,1\}^n,\mu_p^{\otimes n})\to\{0,1\}$ is $\eps$-far from monotone, then
  \[
  \Expect{x\sim\mu_p^{\otimes n}}{\sqrt{s_f^{-}(x)}}\geq \Omega\left(\frac{\eps}{\log(n/\eps)^{2}}\right).
  \]
\end{thm}
\begin{proof}
  Let $r\in\mathbb{N}$ and $p'\in (0,1)$ be from Corollary~\ref{cor:approx_arg} for $\delta = \eps^3$. Note that the distributions
  $\mu_p^{\otimes n}$ and $\mu_{p'}^{\otimes n}$ are $\delta/n^{9}$ close, hence $f$ is at least $\eps/2$ far from monotone over $\mu_{p'}^{\otimes n}$
  and
  \[
  \Expect{x\sim\mu_p^{\otimes n}}{\sqrt{s_f^{-}(x)}}
  \geq
  \Expect{x\sim\mu_{p'}^{\otimes n}}{\sqrt{s_f^{-}(x)}}
  -\sqrt{n}\frac{\delta}{n^{9}}
  \geq
  \frac{\eps}{\sqrt{r}\log(nr/\eps)}
  -\frac{\delta}{n^{8}},
  \]
  where the last inequality is by Lemma~\ref{lem:embed_to_iso}.
  The theorem follows as $r = O(\log^2(n/\eps))$.
\end{proof}

\subsection{Monotone Symmetric Embeddings}\label{sec:thresh_embed}
A function $T\colon \{0,1\}^r \to [m]$ is
monotone and symmetric, if and only if  for each $i\in [m]$, $T^{-1}(i)$ contains
all elements $x$ with hamming weights in some segment, and the segment
that corresponds to $i$ is 'below' that which corresponds to $i+1$ for
each $i$.
Next, we show that if a function $T\colon \{0,1\}^r \to [m]$ is
monotone and symmetric, then the distribution $T(U)$ is $r$-locally
embedded. Here by
$T(U)$ we mean the distribution over $[m]$ resulting from choosing $x$
uniformly  from $\{0,1\}^r$, and outputting $T(x)$.
\begin{lemma}\label{lem:sym_mono}
  Suppose $T\colon \{0,1\}^r \to [m]$ is monotone and symmetric. Then the distribution $T(U)$ is $r$-locally embedded.
\end{lemma}
\begin{proof}
  Denote $\nu = T(U)$. Defining $\phi\colon \{0,1\}^r\to[m]$ by $\phi(x) = T(x)$, it is clear that $\phi$ is monotone and that the distribution of $\phi(U)$ is the same as
  $\nu$, and we next discuss the construction of $\Psi_{\omega}$.

  A monotone path in $\{0,1\}^r$ is a sequence of vertices
  $v_0 = \vec{0},v_1,\ldots,v_r = \vec{1}$ wherein
  $v_0<v_1<\ldots<v_r$ and any two consecutive vertices differ in
  exactly one coordinate. The probability space $(\Omega, P)$ indexes
  a uniform choice of a monotone path in $\{0,1\}^r$ and additional
  auxiliary randomness. One way to generate such path is by choosing a
  random permutation $\pi$ in $S_r$, considering the path going through
  $\vec{0},e_1,e_1+e_2,\ldots, e_1+\ldots+e_i,\ldots,\vec{1}$, and
  applying the permutation to re-label the indices
  $\{1,\ldots,r\}$. We remark that taking a random path
  $\ell = (v_0,\ldots,v_r)$, the marginal distribution of $v_t$ is
  uniform in $\{0,1\}^r$ among all vertices of Hamming weight $t$.

  To define $\Psi_{\omega}$, we look at $\omega$ which specifies a path $\ell = (v_0,\ldots,v_r)$ and additional randomness $\omega'$. We use the additional randomness
  to generate, for each $a\in [m]$, a Hamming weight $t_a$ according to the distribution of $\card{z}$ where we sample $z\sim T^{-1}(a)$ uniformly. We then define
  $\Psi_{\omega}(a) = v_{t_a}$.

  Note that $\phi\circ \Psi_{\omega} = {\sf identity}$, and that for a specific choice of $\omega$, $\Psi_{\omega}(0)\leq \ldots\leq \Psi_{\omega}(m-1)$
  since these are vertices from a monotone path. Finally, fixing $a$, the distribution of $\Psi_{\omega}(a)$ over the randomness of $\omega$ is $v_{t_a}$ where
  $t_a = \card{z}$ and $z\sim T^{-1}(a)$, so $v_{t_a}$ is the $t_a$ vertex on a random monotone path in $\{0,1\}^r$. In other words, $\Psi_{\omega}(a)$ is a uniformly
  chosen vertex from layer $t_a$, where $t_a$ is distributed as above, hence it is uniform in $T^{-1}(a)$. Hence, the distribution of $\Psi_{\omega}(a)$ over
  $\omega\sim P$ and $a\sim \nu$ is uniform over $\{0,1\}^r$.
\end{proof}

Lemma~\ref{lem:sym_mono} can be used to show that distributions that are close to uniform over $[3]$ can be locally embedded. For example, one can choose two thresholds
$t_1<t_2$ and consider the function $T_{t_1,t_2}\colon \{0,1\}^r\to[3]$ defined as $T(x) = 0$ if $\card{x}\leq t_1$, $T(x) = 2$ if $\card{x}\geq t_2$, and otherwise
$T(x) = 1$. A straightforward argument shows that for any $r$, one can choose $t_1,t_2$ so that the distribution $T_{t_1,t_2}(U)$ is $O(1/\sqrt{r})$
close to uniform over $[3]$.
This implies, in particular that as long as
$r\geq \frac{n}{\delta^2}$, the distributions
$T_{t_1,t_2}(U)^{\otimes n}$ and $U^{\otimes n}$ over $[3]$ are
$O(\delta)$-close to each other,\footnote{This can be observes by computing either the KL-divergence or the Hellinger distance between
$T_{t_1,t_2}(U)$ and $U$, which by sub-additivity implies a bound on that measure between $T_{t_1,t_2}(U)^{\otimes n}$ and $U^{\otimes n}$.
One may then conclude a bound on the statistical distance between $T_{t_1,t_2}(U)^{\otimes n}$ and $U^{\otimes n}$ by the relation between
KL-divergence and statistical distance
(via Pinsker's inequality) or by an analogous result for the Hellinger distance.} hence for $\delta<\frac{\eps}{1000}$,
if $f\colon ([3]^n,U^{\otimes n})\to\{0,1\}$ is $\eps$-far from
monotone, then $f\colon ([3]^n,T_{t_1,t_2}(U)^{\otimes n})\to\{0,1\}$
is $\eps/2$-far from monotone, and using Lemma~\ref{lem:tester_conc}
we get a $2$-query monotonicity tester with rejection probability
at least $R(rn, \eps/2)$. A closer inspection shows that the resulting rejection probability is
$\Omega(\eps^2/\sqrt{rn}) = \Omega(\eps^3/n)$ hence worse than known results.

Having said that, the above argument also highlights that if we can design an approximate embedding $T$ such that $T(U)$ is $\xi$-close to
uniform over $[3]$ for $\xi = o(1/\sqrt{r})$, then we will get a non-trivial monotonicity tester over $[3]^n$.
Using elementary arguments, one can show that for any $r$, there is $r' = \Theta(r)$ and thresholds $t_1,t_2$ such that
$T_{t_1,t_2}(U)$ is $O(1/r)$-close to uniform, which allows one to take $r = \Theta(\sqrt{n}/\delta^2)$ and thus get a tester
with rejection probability $\Omega(\eps^3/n^{3/4})$, which already improves upon the state of the art result.

Using threshold as embedding strategy though has its limits. Indeed, it seems that using thresholds we will never be able to get $T(U)$
to be $\xi$-close to uniform over $[3]$ for $\xi = o(1/r^{3/2})$. For each $r'\in [r,100r]$ consider the threshold function $T = T_{t_1,t_2}\colon \{0,1\}^{r'}\to\{0,1,2\}$
that minimizes the distance between $T(U)$ and $U_3$.
Heuristically, one can think of this distance as a random number in the interval $[0,\Theta(1)/\sqrt{r}]$, hence we would expect the minimum
of these to be of the order $1/r^{3/2}$. Thus, to get near optimal monotonicity testers we have to venture beyond threshold functions. In the
the next section we facilitate this by formulating embeddings in the language of monotone perfect matchings (or almost perfect matchings), and
show that slight perturbations of thresholds can be used for embeddings.

\subsection{Embeddings from Monotone Perfect Matchings}\label{sec:match_def}
In this section, we present a combinatorial method of constructing embeddings using monotone matchings on the hypercube.
For simplicity, we tailor our presentation for uniform measures, however one may consider analogs for other distributions.

We will think of the hypercube $G = (\{0,1\}^r, E)$ as a directed graph, wherein $(x,y)$ is an edge if $x<y$.
We may thus view any $\phi\colon \{0,1\}^r\to [m]$ as defining a partitioning of the vertices into sets
$A_0,\ldots,A_{m-1}$ where $A_i = \sett{x}{\phi(x) = i}$.
\begin{definition}\label{def:apm}
  For $\delta>0$, we say a function $\phi\colon \{0,1\}^r\to [m]$ admits a $\delta$-almost perfect matching if
  there are matchings $E_0,\ldots,E_{m-2}$ in $G$, wherein $E_i$ is a matching between $A_i$ and $A_{i+1}$,
  such that for each $i$, $E_i$ covers  all but $\delta$ fraction of the vertices of $A_i$ and $A_{i+1}$.

  If $\phi$ admits a $\delta$-almost perfect matching for $\delta = 0$, we simply say that $\phi$ admits a perfect matching.
\end{definition}

The following lemma asserts that a monotone function $\phi$ that admits a perfect matching can be used toward constructing an
embedding of $([m],U)$.
\begin{lemma}\label{lem:PM_to_embed}
  Let $m\in\mathbb{N}$ and let $\phi\colon \{0,1\}^r\to [m]$ be a monotone function.
  If $\phi$ admits a perfect matching, then $([m], U)$ can be $r$-locally embedded.
\end{lemma}
\begin{proof}
    Let $E_0,\ldots,E_{m-2}$ be monotone matchings in $G$ that cover all vertices for $\phi$, and consider the collection $\mathcal{P}$ of vertex disjoint
    paths of length $m-1$ they form. I.e., starting from a vertex $x\in A_0$ we use the matching edge of $x$ from $E_0$ to go to a vertex $A_1$, then use
    the edge of $E_1$ to go upwards and so on, until we reach $A_{m-1}$. We construct an embedding $(\phi,\Psi=(\Psi_{\omega})_{\omega\in \Omega},\Omega,P)$, where the probability space $\Omega$ is $\mathcal{P}$ and
    the measure $P$ is uniform over $\Omega$. We define $\Psi_{\omega}(i) = \omega_i$, where $\omega_i$ is the vertex from $A_i$ on the path $\omega$.

    The monotonicity of $\phi$ is clear by assumption and the monotonicity of $\Psi_{\omega}$ follows because $\omega$ is a monotone path. Finally, it is
    clear that $\phi\circ \Psi_{\omega} = {\sf identity}$ and that the distribution of $\Psi_{\omega}(i)$ when choosing $i\in [m]$ uniformly and $\omega\sim P$
    is uniform over $\{0,1\}^r$, as $\mathcal{P}$ is a collection of vertex disjoint paths that covers all of $\{0,1\}^r$.
\end{proof}

In light of Lemma~\ref{lem:PM_to_embed}, it makes sense it should be possible to locally embed $([m], U)$ with good parameters.
Indeed, for $m=4$ we found an $9$-local embedding of $[4]$ using computer search~\cite{Python}, which immediately gives near optimal monotonicity testers
and directed isoperimetric inequalities. For divisibility reasons though, to have a perfect matching $m$ must be a power of $2$, however as we show in subsequent
sections, this is the only limitation that exists. To address the divisibility issues, we need to state an analog of approximate embeddings
and prove analogs of Lemmas~\ref{lem:tester_conc},~\ref{lem:embed_to_iso} and~\ref{lem:PM_to_embed}.

\subsection{Monotonicity Testers and Isoperimetric Inequalities from Almost Perfect Matchings}\label{sec:almost_PM}
To circumvent the divisibility issues we consider a more general version of embeddings, which is nevertheless sufficient for the purposes of monotonicity
testing as well as for proving isoperimetric inequalities:
\begin{definition}\label{def:embed2}
  We say that a probability distribution $([m], \mu_1)$ can be $r$-locally embedded in $(\{0,1\}^r, \mu_2)$ if
   there are a map $\phi\colon \{0,1\}^r\to [m]$, a collection of maps $\Psi = \{\Psi_{\omega}\colon [m]\to\{0,1\}^r\}_{\omega\in \Omega}$
  and a probability distribution $P$ over $\Omega$ such that:
  \begin{enumerate}
    \item Each one of $\phi$ and $\Psi_{\omega}$ are monotone.
    \item Sampling $x\sim \mu_2$, the distribution of $\phi(x)$ is $\mu_1$.
    \item Sampling $y\sim \mu_1$ and $\omega\sim P$, the distribution of $\Psi_{\omega}(y)$ is $\mu_2$.
    \item For each $\omega\in \Omega$, $\phi\circ \Psi_{\omega}$ is
        the identity on $[m]$.
  \end{enumerate}
\end{definition}
Definition~\ref{def:embed2} generalizes Definition~\ref{def:embed} in the sense that now we allow the distribution over the hypercube $\{0,1\}^r$ to not be uniform.
In all consequent applications of Definition~\ref{def:embed2} the distribution $\mu_2$ will be very close to uniform, though.
We now prove analogs of Lemmas~\ref{lem:tester_conc},~\ref{lem:embed_to_iso} and~\ref{lem:PM_to_embed} for relaxed embeddings. We begin by showing that almost perfect matchings imply
local embeddings as per Definition~\ref{def:embed2}:
\begin{lemma}\label{lem:relaxed}
  Let $m\in\mathbb{N}$ and $\delta>0$, and let $\phi\colon \{0,1\}^r\to [m]$ be a monotone function.
  If $\phi$ admits a $\delta$-almost perfect matching, then there are distributions $\mu_1$ over $[m]$
  and $\mu_2$ over $\{0,1\}^r$, such that $\mu_1$ is $m\delta$-close to uniform over $[m]$,
  $\mu_2$ is $m\delta$-close to uniform\footnote{In fact, $\mu_1$ is the uniform distribution over a subset of $\{0,1\}^r$ of fractional size
   at least $1-2m\delta$.} over $\{0,1\}^r$ and $([m], \mu_1)$ can be $r$-locally embedded
  in $(\{0,1\}^r, \mu_2)$.
\end{lemma}
\begin{proof}
    We repeat the same construction in Lemma~\ref{lem:PM_to_embed}, except that now the collection $\mathcal{P}$ may include paths of length less than $m-1$.
    We take $\mathcal{P}'\subseteq \mathcal{P}$ to be the collection of paths of length $m-1$. We argue that $\mathcal{P}'$ covers at least
    $1-m\delta$ fraction of vertices of $G$. To see that, note that each path in $\mathcal{P}$ whose length is shorter than $m-1$ can be uniquely
    associated with some $i=0,\ldots,m-2$ and a vertex $x$ either from $A_i$ or $A_{i+1}$ that is not matched in $E_i$.
    Thus, the number of paths in $\mathcal{P}$ shorter than $m-1$ is at most the total number of $(i,x)$ such that $x\in A_i$ is unmatched in $E_i$
    plus the number of $(i,x)$ such that $x\in A_{i+1}$ is unmatched in $E_{i}$, which is at most $2\delta$ fraction of $\{0,1\}^r$. Since each such
    path contains at most $m-1$ vertices, it follows that $\mathcal{P}'$ covers all but $1-2(m-1)\delta$ fraction of $\{0,1\}^r$.

    With this in mind, we define the distribution $\mu_2$ over $\{0,1\}^r$ by picking $\ell\in \mathcal{P}'$ uniformly, $j\in [m]$ uniformly
    and outputting the vertex at the $j^{th}$ spot of the path $\ell$, i.e. $\ell_j$.
    The distribution $\mu_1$ over $[m]$ is defined by sampling $x\sim \mu_2$ and outputting $\phi(x)$.
    We also define $(\Omega,P)$ by taking $\Omega = \mathcal{P}'$
    and $P$ to be the uniform distribution over $\Omega$, and take as before $\Psi = (\Psi_{\omega})_{\omega\in \Omega}$ defined as $\Psi(j) = \omega_j$.

    By definition, the distribution over $\Psi_{\omega}(j)$ where $j\sim \mu_1$ and $\omega\sim \mathcal{P}'$ is $\mu_1$, and the distribution of $\phi(x)$
    where $x\sim \mu_2$ is $\mu_1$. The monotonicity of $\phi, \Psi_{\omega}$ is clear as before, as well as the fact that $\phi\circ \Psi_{\omega} = {\sf identity}$.
\end{proof}

We now turn to the analog of Lemmas~\ref{lem:tester_conc},~\ref{lem:embed_to_iso}.
\begin{lemma}\label{lem:tester_conc2}
  There is an absolute constant $c>0$ such that the following holds.
  Let $r,m,n\in\mathbb{N}$, $\eps,\delta>0$ and suppose that $0<\delta<\frac{c\eps}{mr^2n^2}$.
  If there is a monotone function $\phi\colon\{0,1\}^r\to [m]$ that admits a $\delta$-almost perfect matching, then:
  \begin{enumerate}
    \item there is a $2$-query monotonicity testing algorithm for functions over $([m]^{n},U^{\otimes n})$ that always accepts monotone functions,
    and rejects functions that are $\eps$-far from monotone with probability at least $R(rn,\eps/4)$ (recall~\eqref{eq:R} for the definition of $R$).
    \item
     If $f\colon[m]^n\to\{0,1\}$ is $\eps$-far from monotone with respect to $U^{\otimes n}$, then
     \[
     \Expect{x\in [m]^n}{\sqrt{s_f^{-}(x)}}\geq \Omega\left(\frac{\eps}{\sqrt{r}\log(rn/\eps)}\right).
     \]
  \end{enumerate}

\end{lemma}
\begin{proof}
  Let $\mu_1$ and $\mu_2$ be the distributions from Lemma~\ref{lem:relaxed} from $\phi$, and let $(\phi,(\Psi)_{\omega\in \Omega}, P)$ be
  an $r$-local embedding of $([m],\mu_1)$ in $(\{0,1\}^r,\mu_2)$.
  Given $f\colon ([m]^n,U^{\otimes n})\to\{0,1\}$, define $g\colon\{0,1\}^{r\cdot n}\to\{0,1\}$ by
  \[
    g(x(1),\ldots,x(n)) = f(\phi(x(1)),\ldots,\phi(x(n))).
  \]

  To prove the first item, run the monotonicity tester of the hypercube on $g$, and accept/reject accordingly.
  Note that a single query to $g$ can be answered by making a single query to $f$, and that if $f$ is monotone then $g$ is monotone, hence the tester always
  accepts in this case. If $\eps(f; U^{\otimes n})\geq \eps$, then
  $\eps(f; \mu_1^{\otimes n})\geq \eps(f) - \Delta(\mu_1^{\otimes n}, U^{\otimes n})\geq \eps - 2mn\delta > \eps/2$.
  By the same argument as in Lemma~\ref{lem:pres_dist}, it follows that $g\colon (\{0,1\}^{rn},\mu_2^{\otimes n})\to\{0,1\}$ is $\eps/2$-far from monotone,
  and so $\eps(g; U^{rn})\geq \eps(g;\mu_2^{\otimes n}) - \Delta(U^{rn},\mu_2^{\otimes n})\geq \eps/2 - 2mrn\delta \geq \eps/4$,
  hence the tester rejects with probability at least $R(rn,\eps/4)$.

  To prove the second item, we note that
  \begin{align*}
  \Expect{x\in [m]^n}{\sqrt{s_f^{-}(x)}}
  &\geq
  \Expect{x\sim \mu_1^{\otimes n}}{\sqrt{s_f^{-}(x)}} - \sqrt{n}\Delta(\mu_1^{\otimes n}, U^{\otimes n})\\
  &\geq
  \Expect{y\sim \mu_2^{\otimes n}}{\sqrt{\frac{s_g^{-}(y)}{r}}} - \sqrt{n}\Delta(\mu_1^{\otimes n}, U^{\otimes n})\\
  &=
  \frac{1}{\sqrt{r}}\Expect{y\sim \mu_2^{\otimes n}}{\sqrt{s_g^{-}(y)}} - \sqrt{n}\Delta(\mu_1^{\otimes n}, U^{\otimes n})\\
  &\geq
  \frac{1}{\sqrt{r}}\Expect{y\in \{0,1\}^{rn}}{\sqrt{s_g^{-}(y)}} -\sqrt{rn}\Delta(\mu_2^{\otimes n}, U^{\otimes rn})
  - \sqrt{n}\Delta(\mu_1^{\otimes n}, U^{\otimes n})\\
  &\geq
  \frac{c}{\sqrt{r}}\frac{\eps}{4\log(rn)} -\sqrt{rn}\Delta(\mu_2^{\otimes n}, U^{\otimes rn})
  - \sqrt{n}\Delta(\mu_1^{\otimes n}, U^{\otimes n}),
  \end{align*}
  where $c>0$ is an absolute constant; in the last transition, we used Theorem~\ref{thm:directed_tal_cube}.
  Bounding $\Delta(\mu_2^{\otimes n}, U^{\otimes rn})\leq 2rmn\delta$ and
  $\Delta(\mu_1^{\otimes n}, U^{\otimes n})\leq 2mn\delta$ and using the upper bound on $\delta$ shows that the second and third terms are negligible
  compared to the first, hence we get that
  $\Expect{x\in [m]^n}{\sqrt{s_f^{-}(x)}}\geq \frac{c\eps}{8\sqrt{r}\log(rn)}$ as required.
\end{proof}

With these lemmas in hand, to prove Theorems~\ref{thm:2},~\ref{thm:directed_isoperimetric}
it now suffices to construct good enough almost perfect matchings for some monotone function $\phi\colon \{0,1\}^r\to [m]$.
The following result asserts that such almost perfect monotone matchings exists:
\begin{thm}\label{thm:main_matchings_monotone}
  There is an absolute constant $C>0$ such that for all $m\in\mathbb{N}$, for any $r \geq C\cdot m^6$ there is a monotone function
  $\phi\colon \{0,1\}^r \to [m]$ that admits a $\delta$-almost perfect matching for $\delta\leq m 2^{-r}$.
\end{thm}
The proof of Theorem~\ref{thm:main_matchings_monotone} is deferred to Section~\ref{sec:mono_match_construct}. Before embarking on this proof,
we quickly show how it implies several results stated in the introduction.
\begin{lemma}\label{lem:finish}
  Theorem~\ref{thm:main_matchings_monotone} implies Theorems~\ref{thm:2},~\ref{thm:directed_isoperimetric}.
\end{lemma}
\begin{proof}
Take $r = Cm^6\log(mn/\eps)$ for sufficiently large absolute constant $C>0$. By Theorem~\ref{thm:main_matchings_monotone} we get
that there is $\phi\colon \{0,1\}^r \to [m]$ that admits a $\delta$-almost perfect matching for $\delta\leq m 2^{-r}$,
and the result is concluded by appealing to Lemma~\ref{lem:tester_conc2}.
\end{proof}
\section{Constructing Efficient Monotone Matchings on the Hypercube}\label{sec:mono_match_construct}
\subsection{Theorem~\ref{thm:main_matchings_monotone}: Proof Overview}
We start from a threshold embedding as in Section~\ref{sec:thresh_embed}, that is
$T = T_{t_1,\ldots,t_{m-1}}\colon \{0,1\}^r\to[m]$ defined as $T(x) = i$ if $t_{i}\leq \card{x}< t_{i+1}$.
Using it, we can make sure that the partition it defines, $A_i = \sett{x}{T(x) = i}$ is $\delta$-almost perfect matching for $\delta = O(1/\sqrt{r})$.
The reason for this $\delta$ is that $A_i$'s may have sizes which differ by $2^r\delta$, as this is the number of points in each slice.
Therefore, to improve upon this construction a natural idea is shift elements around by adding to some $A_i$’s elements either from the bottom level of $A_{i+1}$
or from top level of $A_{i-1}$, so that eventually the sizes of all $A_i$’s are equal up to $1$.
We do not know though
how to carry out this adjustment so that the embedding construction from Section~\ref{sec:thresh_embed} still works.
Instead, we vary the sets $A_i$ in a randomized way, and show that with high probability
there is an almost perfect monotone matching between each $A_i$ and $A_{i+1}$ for all $i$'s.

In more details, consider a random ordering $\pi$ of $\{0,1\}^n$ which starts with some ordering of $\{0,1\}^n$ according to Hamming weight
(that is, the vertices of Hamming weight $i$ appear in a chunk before the vertices of Hamming weight $i+1$, for all $i$), and within each
Hamming weight chunk applies a random ordering. We think of $\pi$ as $\pi\colon \{0,1\}^n\to [2^n]$, wherein $\pi^{-1}(i)$ is the $i$th point
in the order. We then define, for each $i=1,\ldots,m$, the set $A_i$ as the chunk of $\left\lfloor \frac{2^n}{m}\right\rfloor$ next elements
in $\pi$, namely
\[
  P_i = \pi^{-1}\left(\sett{\left\lfloor \frac{2^n}{m}\right\rfloor(i-1) + j}{j=1,\ldots,\left\lfloor \frac{2^n}{m}\right\rfloor}\right),
\]
and show that, with high probability, there is a monotone matching between each $P_i$ and $P_{i+1}$.

To show that, we first develop a bit of machinery. First, we generalize the
notion of perfect matching to that of fractional perfect matching (see
Definition~\ref{def:frac_matching}): a fractional perfect matching can
be defined between two sets of equal size, in which case it is
promised that it can be replaced by a true matching. But it can also
be defined over two sets that are each endowed with a arbitrary measure, as
long as the total measure of each set is the same. Another important
property is that the existence of a fractional perfect matching is
transitive, namely if there is a perfect matching between $\mu$ and
$\nu$, and also between $\nu$ and $\tau$ (where $\mu$, $\nu$ and
$\tau$ are sets endowed with measures), then there exists a perfect
matching between $\mu$ and $\tau$.

Then, we view $P_i$ as a collection of $t$ slices,
$s_1(i),\ldots,s_t(i)$ and two random subsets $S_0(i)$ and
$S_{t+1}(i)$ of the slices $s_0(i)$ and $s_{t+1}(i)$. We show that
there is a perfect fractional matching between the vertices of
$\sett{x}{x\in S_0(i)\text{ or $x$ has Hamming weight
  }s_1(i),\ldots,s_{t/2}(i)}$, and vertices of Hamming weight
$s_{t}(i)$ (when weighted appropriately). In words, this says that we
can find a fractional matching between a union of layers with a random
subset of another subsequent layer, and a layer that is a bit above
them.  Using the same arguments, we prove that there is a perfect fractional matching
between the vertices of $\sett{x}{x\in S_{t+1}(i)\text{ or $x$ has Hamming weight
  }s_{t/2+1}(i),\ldots,s_{t}(i)}$ and Hamming weight
$s_{t/4}(i+1)$ vertices.  Thus, in effect we are reduced to matching complete slices
again; indeed, to show the matching between $P_i$ and $P_{i+1}$ we
break them into ``lower half'' and ``upper half'' and use the above
statements to find matchings of these with slices a bit above them and
a bit below them. Using transitivity now and the fact there are
perfect fractional matchings between $s_{t/2}(i)$ and $s_0(i+1)$
(which exists as we make sure that $s_{t/2}(i) < s_0(i+1)$), and
$s_{t/4}(i+1)$ and $s_{t/2}(i+1)$ (which again exists as we make sure
that $s_{t/4}(i+1) < s_{t/2}(i+1)$), one can then stitch these
matchings to get a perfect fractional matching between $P_i$ and
$P_{i+1}$, and thus conclude the existence of perfect matching.

The proof of statements such that ``there is a perfect fractional
matching between $S_0(i)$ and the slices $s_1(i),\ldots,s_{t/2}(i)$
and $s_{t/2}(i)$ (when weighted appropriately)'' consists the bulk of
the work, and to do that we show that with high probability Hall's
condition holds.  To do that, we use the notion of upper shadows
(which, roughly speaking, counts the number of neighbours a set of
vertices $S$ has in the directed hypercube graph) as well as the
Kruskal-Katona theorem which gives us a lower bound on it. We show
that only sets of vertices $T$ which have very smaller upper shadow
can violate Hall's condition, and for them we show by a careful
application of Chernoff's bound that, with high probaiblity, Hall's
condition still holds. The main difficulty in the last step is that
the number of such sets $T$ is quite large, however we show that these
sets admit an efficient ``$\eps$-net'' type approximations. This
reduces the number of sets $T$ we need to union bound over enough so
that Chernoff's bound works.
\subsection{Shadows, Kruskal-Katona and Approximating Collections with Small Shadow}
\subsubsection{The Kruskal-Katona Theorem}
Throughout this section, we consider slices of the Boolean hypercube, $\binom{[n]}{k} = \sett{x\in \{0,1\}^n}{\card{x} = k}$, and denote by
$\mu_k$ the uniform measure on $\binom{[n]}{k}$.
Our proof uses the Kruskal-Katona Theorem~\cite{katona2009theorem,kruskal1963number,bollobas1987threshold}, which we present next.
We will use a more convenient form of it as stated in~\cite[Section 1.2]{o2013kkl}.
\begin{definition}
  For a collection $\mathcal{A}\subseteq \binom{[n]}{k}$, define the upper shadow $\partial^u\mathcal{A}$ and lower shadow
  $\partial^d \mathcal{A}$ of $\mathcal{A}$ as
  \[
    \partial^u\mathcal{A} = \sett{y\in\binom{[n]}{k+1}}{\exists x\in\mathcal{A}, x<y},
    \qquad
    \partial^d\mathcal{A} = \sett{y\in\binom{[n]}{k-1}}{\exists x\in\mathcal{A}, y<x}.
  \]
\end{definition}

The Kruskal-Katona Theorem states:
\begin{lemma}\label{lemma:KK}
  For all $\mathcal{A}\subseteq \binom{[n]}{k}$ we have that
  \[
    \mu_{k+1}(\partial^u\mathcal{A})
    \geq
    \mu_{k}(\mathcal{A})^{1-\frac{1}{n}},
    \qquad\qquad
    \mu_{k-1}(\partial^d\mathcal{A})
    \geq
    \mu_{k}(\mathcal{A})^{1-\frac{1}{n}}.
  \]
\end{lemma}

In words, Lemma~\ref{lemma:KK} asserts that if $\mathcal{A}$ is a small sub-set of a slice, then the upper shadow (as well as the lower shadow) have
considerably larger densities. Typically, we will apply the upper shadow/ lower shadow operators more than once; given $\mathcal{A}\subseteq\binom{[n]}{k}$,
we will look at $\mu_{k+t}(\partial^{u}\ldots\partial^u \mathcal{A})$ where we applied the upper shadow operator $t$-times. To simplify notations,
we denote this by $\mu_{k+t}(\partial^{t\cdot u} \mathcal{A})$,

\subsubsection{Approximating a Collection with a Small Shadow}
In general, the conclusion of Lemma~\ref{lemma:KK} is tight, as can be evidenced by collections of the type
\[
\mathcal{A} = \sett{ x\in\binom{[n]}{k}}{x_1=\ldots=x_{\ell} = 1}.
\]
Intuitively, the reason that $\mathcal{A}$ above is tight for Kruskal-Katona is that for any element in $y\in \binom{[n]}{k+1}$, we either have that almost all of the $x<y$
of Hamming weight $k$ are in $\mathcal{A}$ -- in which case $y\in \partial^{u} \mathcal{A}$, or else none of these $x$'s are in $\mathcal{A}$. Hence, many of the $x$'s ``vote'' for
the same set of $y$'s to be included in the upper shadow, leading to only a moderate increase in density. We show that in general, collections $\mathcal{A}$ with small
shadow exhibit such behaviour, and use it to show that this collection of families admits a small $\eps$-net:
\begin{lemma}\label{lem:sparse_approx}
  Let $s,t,n\in\mathbb{N}$ such that $0\leq t\leq n-s$, and let $0<\eps\leq \frac{1}{100}$. If $\mathcal{A}\subseteq \binom{[n]}{s}$ satisfies
  $\mu_{s+t}(\partial^{t\cdot u} \mathcal{A})\leq (1+\eps)\mu_s(\mathcal{A})$, then there is a collection
  $\mathcal{M}\subseteq \mathcal{A}$ and $\mathcal{B}_{\mathcal{M}} \subseteq \binom{[n]}{s}$, $\mathcal{B}_{\mathcal{M}}'\subseteq \binom{[n]}{s+t}$
  (defined only by $\mathcal{M}$)
  such that
  \begin{enumerate}
    \item $\card{\mathcal{M}}\leq 100\frac{\ln(1/\eps)}{\binom{s+t}{t}}\cdot \card{\mathcal{A}}$.
    \item $\mathcal{B}_{\mathcal{M}}' = \partial^{t\cdot u} \mathcal{M}$ and
    $\mu_{s+t}(\mathcal{B}_{\mathcal{M}}'\Delta \partial^{t\cdot u} \mathcal{A})\leq 6\eps \cdot\mu_{s+t}(\partial^{t\cdot u} \mathcal{A})$.
    \item $\mathcal{B}_{\mathcal{M}} = \sett{x\in \binom{[n]}{s}}{\Prob{y>x, \card{y} = s+t}{y\in \mathcal{B}_{\mathcal{M}}'}\geq \frac{1}{2}}$ and
    $\mu_s(\mathcal{B}_{\mathcal{M}}\Delta \mathcal{A})\leq 18\eps \mu_{s}(\mathcal{A})$.
  \end{enumerate}
\end{lemma}
\begin{proof}
  We show that taking $\mathcal{M}\subseteq \mathcal{A}$ randomly of size $M = 100\card{\mathcal{A}} \frac{\ln(1/\eps)}{\binom{s+t}{t}}$, the collections
  $\mathcal{B}_{\mathcal{M}}$ and $\mathcal{B}_{\mathcal{M}}'$ as defined in the statement work with positive probability.

  Consider the bi-partite graph $G = (V\cup U, E)$ where the sides are $V = \mathcal{A}$ and $U = \partial^{t\cdot u} \mathcal{A}$,
  and $(x,y)\in E$ is an edge if $x\in\mathcal{A}$, $y\in \partial^{t\cdot u} \mathcal{A}$ and $x<y$. Then $G$ is left-regular with degree
  $h_{L} = \binom{n-s}{t}$, and so
  \[
  \card{E} = \card{V}\cdot \binom{n-s}{t} = \mu_s(\mathcal{A})\binom{n-s}{t}\binom{n}{s} = \mu_s(\mathcal{A})\frac{n!}{s!\cdot t!\cdot (n-s-t)!}.
  \]
  As for the right side, the degree of each vertex is at most $h_R = \binom{s+t}{t}$ and the average degree of a vertex is
  \[
  \frac{\card{E}}{\card{U}} = \frac{\mu_s(\mathcal{A})\frac{n!}{s!\cdot t!\cdot (n-s-t)!}}{\card{U}} = \frac{\mu_s(\mathcal{A})\binom{n}{s+t} h_R}{\card{U}}
  =h\frac{\mu_s(\mathcal{A})}{\mu_{s+t}(\mathcal{A})}
  \geq \frac{h_R}{1+\eps}.
  \]
  Thus, choosing $y\in U$ uniformly, the expected value of $h_R - d(y)$ is at most $\eps h_R$, and by Markov's inequality it follows that $h_R-d(y)\leq h_R/2$
  except with probability $2\eps$. Thus, denoting by $\delta$ the fraction of $y\in U$ such that $d(y)<h_R/2$, we get that $\delta\leq 2\eps$. Thus,
  \[
        \Expect{\mathcal{M}}{\mu_{s+t}(\mathcal{B}_{\mathcal{M}}'\Delta \partial^{t\cdot u} \mathcal{A})}
        \leq \delta \mu_{s+t}(\partial^{t\cdot u} \mathcal{A}) + (1-\delta)\mu_{s+t}(\partial^{t\cdot u} \mathcal{A})\left(1-\frac{h_R/2}{\card{V}}\right)^{M}
        \leq 3\eps \cdot\mu_{s+t}(\partial^{t\cdot u} \mathcal{A}),
  \]
  where in the last inequality we used the fact that $\left(1-\frac{h_R/2}{\card{V}}\right)^{M}
  \leq
  e^{-\frac{M h_R}{2\card{V}}}
  \leq
  e^{-50\ln(1/\eps)}
  \leq \eps$.

  The third item follows using a similar argument, and we first upper bound $\Expect{\mathcal{M}}{\mu_s(\mathcal{A}\setminus \mathcal{B}_{\mathcal{M}})}$.
  For each $x\in \mathcal{A}\setminus\mathcal{B}_{\mathcal{M}}$ we have that $x$ has at most $h_L/2$ of its neighbours in $\mathcal{B}_{\mathcal{M}}'$, hence at least
  $h_L/2$ of its neighbours in $U\setminus \mathcal{B}_{\mathcal{M}}'$. It follows that
  \[
  \Expect{\mathcal{M}}{\mu_s(\mathcal{A}\setminus \mathcal{B}_{\mathcal{M}})}
  \leq
  \frac{1}{\binom{n}{s}}\Expect{\mathcal{M}}{\frac{\card{U\setminus \mathcal{B}_{\mathcal{M}}'}h_R}{h_L/2}}
  =\frac{2h_R\binom{n}{s+t}}{h_L\binom{n}{s}}\Expect{\mathcal{M}}{\mu_{s+t}(\mathcal{B}_{\mathcal{M}}'\Delta \partial^{t\cdot u} \mathcal{A})}
  \leq 6\eps \cdot\mu_{s+t}(\partial^{t\cdot u} \mathcal{A}),
  \]
  which is at most $7\eps \mu_s(\mathcal{A})$.
  To upper bound $\Expect{\mathcal{M}}{\mu_s(\mathcal{B}_{\mathcal{M}}\setminus \mathcal{A})}$, we note that any $x\in \mathcal{B}_{\mathcal{M}}\setminus \mathcal{A}$
  has at least $h_L/2$ of the $y$ of Hamming weight $s+t$ for which $x<y$ in $\mathcal{B}_{\mathcal{M}}'$,
  and in particular in $U$. The total number of pairs $(x,y)$ such that $x<y$ and $x\not\in\mathcal{A}$, $y\in U$ is at most $h_R\card{U} - \card{E}$
  (as these are all non-edges in $G$), so we get that
  \[
  \Expect{\mathcal{M}}{\mu_s(\mathcal{B}_{\mathcal{M}}\setminus \mathcal{A})}
  \leq \frac{1}{\binom{n}{s}}\frac{h_R\card{U} - \card{E}}{h_L/2}
  \leq \frac{1}{\binom{n}{s}}\frac{2h_R\card{U}\eps}{h_L(1+\eps)}
  \leq 2\eps \frac{\mu_{s+t}(U)}{1+\eps}
  \leq 2\eps \mu_s(\mathcal{A}).
  \]

  In conclusion, we get that
  \[
  \Expect{\mathcal{M}}{\mu_s(\mathcal{B}_{\mathcal{M}}\Delta \mathcal{A})}\leq 9\eps\mu_s(\mathcal{A}),
  \qquad
  \Expect{\mathcal{M}}{\mu_{s+t}(\mathcal{B}_{\mathcal{M}}'\Delta \partial^{t\cdot u} \mathcal{A})}\leq 3\eps\mu_{s+t}(\partial^{t\cdot u} \mathcal{A}),
  \]
  so by Markov's inequality there is a choice for $\mathcal{M}$ satisfying the conclusion of the claim.
\end{proof}

For future reference, we state a version of Lemma~\ref{lem:sparse_approx} for the operator $\partial^{t\cdot d}$:
\begin{lemma}\label{lem:sparse_approx2}
  Let $s,t,n\in\mathbb{N}$ such that $0\leq t\leq n-s$, and let $0<\eps\leq \frac{1}{100}$. If $\mathcal{A}\subseteq \binom{[n]}{s}$ satisfies
  $\mu_{s-t}(\partial^{t\cdot d} \mathcal{A})\leq (1+\eps)\mu_s(\mathcal{A})$, then there is a collection
  $\mathcal{M}\subseteq \mathcal{A}$ and $\mathcal{B}_{\mathcal{M}}' \subseteq \binom{[n]}{s-t}$, $\mathcal{B}_{\mathcal{M}}\subseteq \binom{[n]}{s}$
  such that
  \begin{enumerate}
    \item $\card{\mathcal{M}}\leq 100 \frac{\ln(1/\eps)}{\binom{s}{s-t}}\cdot \card{\mathcal{A}}$.
    \item $\mathcal{B}_{\mathcal{M}} ' = \partial^{t\cdot d} \mathcal{M}$ and
    $\mu_{s-t}(\mathcal{B}_{\mathcal{M}}'\Delta \partial^{t\cdot d} \mathcal{A})\leq 6\eps \cdot\mu_{s-t}(\partial^{t\cdot d} \mathcal{A})$.
    \item $\mathcal{B}_{\mathcal{M}} = \sett{x\in \binom{[n]}{s}}{\Prob{y<x, \card{y} = s-t}{y\in \mathcal{B}_{\mathcal{M}}'}\geq \frac{1}{2}}$ and
    $\mu_s(\mathcal{B}_{\mathcal{M}}\Delta \mathcal{A})\leq 18\eps \mu_{s}(\mathcal{A})$.
  \end{enumerate}
\end{lemma}
\begin{proof}
  The proof is essentially the same as the proof of Lemma~\ref{lem:sparse_approx} and we omit the details.
\end{proof}

\subsection{Fractional Monotone Matchings}
We now formally define the concept of a monotone fractional matching, which is central to the proof of Theorem~\ref{thm:main_matchings_monotone}:
\begin{definition}\label{def:frac_matching}
  Let $w_U\colon \{0,1\}^n \to [0,\infty)$ and $w_V\colon \{0,1\}^n\to [0,\infty)$ be weight functions such that $\sum\limits_{x\in \{0,1\}^n} w_U(x) = \sum\limits_{x\in \{0,1\}^n} w_V(x)$.
  We say there is a monotone fractional matching from $w_U$ to $w_V$, and denote $w_U\lesssim w_V$ if, letting $U$ be the support of
  $w_U$ and $V$ be the support of $w_V$, there is a weight function $w\colon U\times V\to[0, \infty)$ such that $w(u,v)>0$ only when $u\leq v$,
  and for every $x\in U$ $y\in V$ it holds that
  \[
  \sum\limits_{z\in V} w(x,z) = w_U(x),
  \qquad
  \sum\limits_{z\in U} w(z,y) = w_V(y).
  \]
\end{definition}
In this section, we establish several basic properties of fractional monotone matchings. The first of which is a fractional version of
Hall's Theorem for monotone matchings. For completeness, we include the (straight-forward) deduction of it from the usual formulation of Hall's Theorem.
\begin{lemma}\label{lem:Hall}
  Suppose that $w_U$ and $w_V$ are as in Definition~\ref{def:frac_matching}, let $U$ and $V$ their supports, and suppose that for all
  $S\subseteq U$, defining $N(S) = \sett{v\in V}{\exists u\in S, u\leq v}$ we have that
  \[
  \sum\limits_{v\in N(S)} w_V(v)\geq \sum\limits_{u\in S} w_U(u).
  \]
  Then there is a monotone fractional matching between $w_U$ and $w_V$.
\end{lemma}
\begin{proof}
  By approximation, it suffices to show that statement for weight functions $w_V$ and $w_U$ that assign rational values.
  Let $M$ be a number such that all values of $M w_V$ and $M w_U$ are integers, and define the bi-partite graph $G$ whose sides are $U'$ and $V'$,
  where each $u\in U$ has $M w_U(u)$ copies in $U'$ and $v\in V$ has $M w_V(v)$ copies in $V'$. We connect $(u',v')$ by an edge if they are copies
  of $u\in U$, $v\in V$ respectively where $u\leq v$. Our assumption then implies that $G$ satisfies Hall's condition, so we may find a perfect
  matching $M \subseteq U'\times V'$. Define
  \[
  w(u,v) = \sum\limits_{\substack{u'\text{ copy of $u$}\\ v'\text{ copy of $v$}}} 1_{(u',v')\in M},
  \]
  and note that then $w$ forms a fractional monotone matching showing $w_U \lesssim w_V$.
\end{proof}

Secondly, we have the following basic properties of $\lesssim$:
\begin{lemma}\label{lem:basic_prop_matchings}
  Suppose that $w_U,w_V,w_R,w_L\colon \{0,1\}^n\to[0,\infty)$ are weight functions.
  \begin{enumerate}
    \item Transitivity: if $w_U\lesssim w_V$ and $w_V\lesssim w_R$, then $w_U\lesssim w_R$.
    \item Linearity: if $w_U\lesssim w_R$ and $w_V\lesssim w_L$, then for all $p,q\geq 0$,
    $p w_U + qw_V\lesssim p w_R + q w_L$.
  \end{enumerate}
\end{lemma}
\begin{proof}
  For the first item, let $U,V,R$ be the supports of $w_U,w_V$ and $w_R$ respectively, and let
  $w_1\colon U\times V\to [0,\infty)$ and $w_2\colon V\times R\to[0,\infty)$ be the weight functions demonstrating that
  $w_U\lesssim w_V$ and $w_V\lesssim w_R$, respectively. Define
  $w\colon U\times R\to [0,\infty)$ by
  \[
  w(u,r) = \sum\limits_{v\in V}\frac{1}{w_V(v)}w_1(u,v)w_2(v,r).
  \]
  First, if $w(u,r)>0$ then there is $v\in V$ such that $w_1(u,v), w_2(v,r)>0$ and so $u\leq v\leq r$, hence $u\leq r$.
  Secondly, note that for all $u$,
  \[
  \sum\limits_{r\in R}w(u,r)
  =\sum\limits_{v\in V}\frac{1}{w_V(v)}w_1(u,v)\sum\limits_{r\in R}w_2(v,r)
  =\sum\limits_{v\in V}\frac{1}{w_V(v)}w_1(u,v)w_V(v)
  =\sum\limits_{v\in V}w_1(u,v)
  =w_U(u),
  \]
  and similarly for all $r\in R$ we have $\sum\limits_{u\in U}w(u,r) = w_R(r)$. It follows that $w$ is a monotone matching between showing that
  $w_U\lesssim w_R$.

  For the second item, let $U,R,V,L$ be the supports of $w_U, w_R, w_V$ and $w_L$ respectively and let
  $w_1\colon U\times R\to [0,\infty)$ and $w_2\colon V\times L \to [0,\infty)$ be weight functions demonstrating that
  $w_U\lesssim w_R$ and $w_V\lesssim w_L$. Then $w(x,y) = pw_1(x,y) + qw_2(x,y)$ is a weight function showing that
  $p w_U + qw_V\lesssim p w_R + q w_L$.
\end{proof}

Third, we show that if $k\leq k'$, then $\mu_k\lesssim \mu_{k'}$.
\begin{lemma}\label{lem:path_dom}
  If $k\leq k'$, then $\mu_k\lesssim \mu_{k'}$.
\end{lemma}
\begin{proof}
  Let $P=(v_1,\ldots,v_n)$ be a uniformly chosen monotone path in $\{0,1\}^n$ starting at $(0,\ldots,0)$ and ending at $(1,\ldots,1)$, and
  define $w(x,y)$ to be the probability that $v_k = x$ and $v_{k'} = y$. Then it is easily seen that $w(x,y)>0$ only if $x<y$, and also
  for every $x$ of Hamming weight $k$, $\sum\limits_{y} w(x,y)$ is equal to the probability a uniformly chosen vertex of Hamming weight
  $k$ is equal to $x$, hence is $\mu_k(x)$. Similarly, $\sum\limits_{x} w(x,y) = \mu_{k'}(y)$.
\end{proof}

The last statement is a standard connection between fractional matchings and perfect matchings.
\begin{lemma}\label{lem:frac_to_matching}
  Suppose that $\mu,\mu'$ are distributions which are uniform over $A,A'\subseteq\{0,1\}^n$ respectively, where $\card{A} = \card{A'}$.
  If $\mu\lesssim \mu'$, then there is a monotone perfect matching between $A$ and $A'$.
\end{lemma}
\begin{proof}
  Consider the bipartite graph $G = (A\cup A', E)$ where $E = \sett{(a,a')}{a\in A, a'\in A', a\leq a'}$.
  As $\mu\lesssim \mu'$, we get that there is $w\colon A\times A'\to[0,\infty)$ supported only on $E$ satisfying
  the properties of a monotone fractional matching. Define $w' = \card{A} w$, and note that for all $a\in A$
  we have that $\sum\limits_{a'\in A'} w'(a,a') = 1$ and also $\sum\limits_{a\in A} w'(a,a') = 1$ for
  all $a'\in A'$. Thus, the fractional matching number of $G$ is at least
  \[
  \sum\limits_{a\in A, a'\in A'} w'(a,a') = \card{A}.
  \]
  We argue that the smallest vertex cover in $G$ has size $\card{A}$. Indeed, if $W\subseteq A\cup A'$ is a vertex cover then
  \[
  \card{A}
  =
  \sum\limits_{e\in E} w'(e)
  \leq
  \sum\limits_{z\in W}\sum\limits_{e\ni z} w'(e) = \sum\limits_{z\in W} 1 = \card{W}.
  \]
  It now follows from K\H{o}nig's theorem that $G$ has a perfect matching, and we are done.
\end{proof}
\subsection{Monotone Matchings on Random Subsets of the Slice}
The next lemma is the heart of the proof that our construction admits a good monotone almost perfect matching.
For a collection $\mathcal{S}\subseteq \binom{[n]}{k}$, we denote $\mu_\mathcal{S}(x) = \mu_k(x) 1_{x\in \mathcal{S}}$.
\begin{lemma}\label{lem:check_hall}
  For all $C>0$ there is $n_0\in\mathbb{N}$ such that the following holds.
  Let $n,k,s,t\in\mathbb{N}$ and assume that $\frac{n}{2}-C\sqrt{n\log n}\leq k\leq \frac{n}{2} + C\sqrt{n\log n}$
  and $10 n^{1/3}\leq t\leq C\sqrt{n\log n}$. Then for every $0\leq s\leq \binom{n}{k}$, setting $\rho = \frac{s}{\binom{n}{k}}$ we have
  \[
        \Prob{\substack{\mathcal{S}\subseteq \binom{[n]}{k}\\ \card{\mathcal{S}} = s}}{(t+\rho)\mu_{k-t}\lesssim t\mu_k + \mu_\mathcal{S}}
        \geq 1-2^{-\Omega(2^{n/2})}.
  \]
\end{lemma}
\begin{proof}
  We will use Lemma~\ref{lem:Hall}. Denoting $\mu = (t+\rho) \mu_{k-t}$ and $\nu = \nu_\mathcal{S} = t\mu_k + \mu_\mathcal{S}$, our goal is to show that with high probability
  over the choice of $\mathcal{S}$, for all $\mathcal{T}\subseteq \binom{[n]}{k-t}$ it holds that $\nu(\partial^{t\cdot u} \mathcal{T})\geq \mu(\mathcal{T})$. Equivalently, we will upper
  bound the probability that there is $\mathcal{T}$ that violates it, and we present two arguments depending on the fractional size of $\mathcal{T}$.

  \paragraph{The case that $\mu_{k-t}(\mathcal{T})\leq 1/2$.}
  Let $\mathcal{T}$ be such that $\mu_{k-t}(\mathcal{T})\leq 1/2$, and suppose that $\nu(\partial^{t\cdot u} \mathcal{T})< \mu(\mathcal{T})$. We denote by $N$ the size of $\mathcal{T}$.
  Then we have that
  \begin{equation}\label{eq:2}
  \mu_{k}(\partial^{t\cdot u} \mathcal{T})\leq \frac{1}{t} \nu(\partial^{t\cdot u} \mathcal{T})\leq \frac{1}{t} \mu(\mathcal{T})\leq \frac{t+1}{t} \mu_{k-t}(\mathcal{T}),
  \end{equation}
  On the other hand, using Lemma~\ref{lemma:KK} we can deduce a lower bound on the measure of the upper shadow of $\mathcal{T}$, namely that
  \begin{equation}\label{eq:applyKK}
  \mu_{k}(\partial^{t\cdot u} \mathcal{T})\geq \mu_{k-t}(\mathcal{T})^{\left(1-\frac{1}{n}\right)^t}.
  \end{equation}
  First, this implies a lower bound on the measure of $\mathcal{T}$,
  as we get that $\frac{t+1}{t} \mu_{k-t}(\mathcal{T})\geq \mu_{k-t}(\mathcal{T})^{\left(1-\frac{1}{n}\right)^t}$, and standard manipulations now imply that
  $\mu_{k-t}(\mathcal{T})\geq (1-1/(t+1))^{n/t}$ and so $\mu_{k-t}(\mathcal{T})\geq e^{-O(n/t^2)}$,
  which implies in particular that $N\geq e^{-O(n/t^2)}\binom{n}{k-t}\geq 2^{0.8 n}$. Secondly, from~\eqref{eq:applyKK} and the fact that $\mu_{k-t}(\mathcal{T})\leq 1/2$ we also get that
  \begin{equation}\label{eq:3}
        \mu_{k}(\partial^{t\cdot u} \mathcal{T})
        \geq \mu_{k-t}(\mathcal{T}) (1/2)^{\left(1-\frac{1}{n}\right)^t-1}\geq \mu_{k-t}(\mathcal{T}) 2^{t/n}.
  \end{equation}
  Combining our assumption on $\mathcal{T}$ and~\eqref{eq:3} yields
  \begin{align*}
  0
  \leq \mu(\mathcal{T}) - \nu(\partial^{t\cdot u} \mathcal{T})
  &= (t+\rho)\mu_{k-t}(\mathcal{T}) - t\mu_k(\partial^{t\cdot u} \mathcal{T}) - \mu_k(\partial^{t\cdot u} \mathcal{T}\cap \mathcal{S})\\
  &\leq (t+\rho)\mu_{k-t}(\mathcal{T}) - t\mu_{k-t}(\mathcal{T}) 2^{t/n}- \mu_k(\partial^{t\cdot u} \mathcal{T}\cap \mathcal{S})\\
  &\leq \left(\rho  - \frac{t^2}{2n}\right)\mu_{k-t}(\mathcal{T})- \mu_k(\partial^{t\cdot u} \mathcal{T}\cap \mathcal{S}),
  \end{align*}
  and in particular we get that
  \begin{equation}\label{eq:4}
  \mu_k(\partial^{t\cdot u} \mathcal{T}\cap \mathcal{S})\leq \left(\rho  - \frac{t^2}{2n}\right)\mu_{k-t}(\mathcal{T})\leq \left(\rho  - \frac{t^2}{2n}\right)\mu_{k}(\partial^{t\cdot u}\mathcal{T}).
  \end{equation}
  Noting that the expectation of the left hand side, over the choice of $\mathcal{S}$, is $\rho \mu_{k}(\partial^{t\cdot u}\mathcal{T})$, this inequality suggests that
  the probability for this for a specific $\mathcal{T}$ is small. A naive application of Chernoff's bound is not good enough since we would need to
  union bound over too many choices for $\mathcal{T}$. To cut down on the number of events we union bound over, we observe that as~\eqref{eq:2} holds
  we may move to a sparse approximator of $\mathcal{T}$ and thus handle much less sets.

  More precisely, using Lemma~\ref{lem:sparse_approx} and the guarantee from~\eqref{eq:2} we get that there is $\mathcal{M}$ of size at most
  $\alpha N$ for $\alpha = \frac{100\log t}{\binom{k}{t}}$ satisfying the conclusion of the lemma for $\mathcal{B}_{\mathcal{M}}$ and $\mathcal{B}_{\mathcal{M}}'$ as therein. It follows that
  \[
  \mu_k(\mathcal{B}_{\mathcal{M}}'\cap \mathcal{S})
  \leq
  \mu_k(\partial^{t\cdot u} \mathcal{T}\cap \mathcal{S})
  +
  \mu_k(\partial^{t\cdot u} \mathcal{T} \Delta \mathcal{B}_{\mathcal{M}}')
  \leq \left(\rho - \frac{t^2}{2n} + \frac{6}{t}\right)\mu_k(\partial^{t\cdot u} \mathcal{T})
  \leq \left(\rho - \frac{t^2}{2n} + \frac{7}{t}\right)\mu_k(\mathcal{B}_{\mathcal{M}}'),
  \]
  where in the last inequality we used the fact that $\mu_k(\partial^{t\cdot u} \mathcal{T})\leq \left(1+\frac{1}{t}\right)\mu_k(\mathcal{B}_{\mathcal{M}}')$ by the conclusion of  Lemma~\ref{lem:sparse_approx}.
  By the condition on $t$, $\rho - \frac{t^2}{2n} + \frac{7}{t}\leq \rho - \frac{t^2}{3n}$ and hence we conclude that
  \[
    \mu_k(\mathcal{B}_{\mathcal{M}}'\cap \mathcal{S})\leq \left(\rho - \frac{t^2}{3n}\right)\mu_k(\mathcal{B}_{\mathcal{M}}').
  \]
  We may assume that $\rho\geq t^2/3n$, otherwise the last inequality is impossible. We also note that from Lemma~\ref{lem:sparse_approx}
  we have $\mu_{k}(\mathcal{B}_{\mathcal{M}}')\geq \frac{1}{2}\mu_{k}(\partial^{t\cdot u} \mathcal{T})\geq \frac{1}{2}\mu_{k-t}(\mathcal{T}) = \frac{N}{2\binom{n}{k-t}}$. From everything claimed so far we conclude that
  \begin{align}
  &\Prob{S}{\exists \mathcal{T} \text{ with $\card{\mathcal{T}} = N$ such that }\nu(\partial^{t\cdot u} \mathcal{T})< \mu(\mathcal{T})} \notag\\
  &\leq
  \Prob{S}{\exists \mathcal{M}\text{ with $\card{\mathcal{M}}\leq \alpha N$ such that $\mu_{k}(\mathcal{B}_{\mathcal{M}}')\geq \frac{N}{2\binom{n}{k-t}}$ and }
  \mu_k(\mathcal{B}_{\mathcal{M}}'\cap S)\leq \left(\rho - \frac{t^2}{3n}\right)\mu_k(\mathcal{B}_{\mathcal{M}}')} \notag\\\label{eq:6}
  &\leq\hspace{-5ex}
  \sum\limits_{\substack{\mathcal{M}\subseteq \binom{n}{k} \\ \card{\mathcal{M}}\leq \alpha N \\ \mu_{k}(\mathcal{B}_{\mathcal{M}}')\geq \frac{N}{2\binom{n}{k-t}}}}
  \Prob{S}{\mu_k(\mathcal{B}_{\mathcal{M}}'\cap S)\leq \left(\rho - \frac{t^2}{3n}\right)\mu_k(\mathcal{B}_{\mathcal{M}}')}.
  \end{align}
  For each $\mathcal{M}$ such that $\card{\mathcal{M}}\leq \alpha N$ and $\mu_{k}(\mathcal{B}_{\mathcal{M}}')\geq \frac{N}{2\binom{n}{k-t}}$, let $E_{\mathcal{M}}$ be the event that $\mu_k(\mathcal{B}_{\mathcal{M}}'\cap \mathcal{S})\leq \left(\rho - \frac{t^2}{3n}\right)\mu_k(\mathcal{B}_{\mathcal{M}}')$;
  we upper bound the probability of each $E_{\mathcal{M}}$ separately, and for that we use Chernoff's bound. There is a slight technical issue in applying Chernoff's bound,
  namely that $\mathcal{S}$ is selected to be of fixed size, and to circumvent it we consider $\mathcal{S}' \subseteq \binom{[n]}{k}$ chosen randomly by including each set
  from $\binom{[n]}{k}$ in it with probability $\rho' = \rho - t^2/6n$. Then we get that
  \[
  \Prob{}{\card{\mathcal{S}'} > s}=\Prob{}{\card{\mathcal{S}'} > \left(\rho'+\frac{t^2}{6n}\right)\binom{n}{k}}\leq e^{-\Omega\left(\frac{t^2}{n^4}\rho\binom{n}{k}\right)}
  \leq 0.5,
  \]
  where we used Chernoff's bound and $\rho\geq t^2/3n$. Note that $\Prob{\mathcal{S}}{\mathcal{E}_{\mathcal{M}}}\leq \cProb{\mathcal{S}'}{\card{\mathcal{S}'}\leq s}{\mathcal{E}_{\mathcal{M}}}$,
  and combining with the above bound on the probability that $\card{\mathcal{S}'}>s$ we get that
  \begin{align*}
  \Prob{\mathcal{S}}{\mathcal{E}_{\mathcal{M}}}
  \leq 2\Prob{\mathcal{S}'}{\mathcal{E}_{\mathcal{M}}}
  &=2\Prob{\mathcal{S}'}{\mu_k(\mathcal{B}_{\mathcal{M}}'\cap \mathcal{S}')\leq \left(\rho - \frac{t^2}{3n}\right)\mu_k(\mathcal{B}_{\mathcal{M}}')}\\
  &=2\Prob{\mathcal{S}'}{\mu_k(\mathcal{B}_{\mathcal{M}}'\cap \mathcal{S}')\leq \left(\rho' - \frac{t^2}{6n}\right)\mu_k(\mathcal{B}_{\mathcal{M}}')}\\
  &\leq 2e^{-\Omega\left(\frac{t^4}{n^2} \rho'\mu_k(\mathcal{B}_{\mathcal{M}}')\binom{n}{k}\right)}\\
  &\leq 2e^{-\Omega\left(\frac{t^4}{n^2} \rho'\frac{N}{2\binom{n}{k-t}} \binom{n}{k}\right)}\\
  &\leq e^{-\Omega\left(\frac{t^6}{n^3} \frac{N\binom{n}{k}}{\binom{n}{k-t}} \right)}.
  \end{align*}
  Thus, using~\eqref{eq:6} we get that the left hand side therein is upper bounded by
  \begin{align*}
  \sum\limits_{\substack{\mathcal{M}\subseteq \binom{n}{k} \\ \card{\mathcal{M}}\leq \alpha N}} e^{-\Omega\left(\frac{t^6}{n^3} \frac{N\binom{n}{k}}{\binom{n}{k-t}} \right)}
  \leq 2^{n\alpha N -\Omega\left(\frac{t^6}{n^3} \frac{N\binom{n}{k}}{\binom{n}{k-t}} \right)}
  \leq 2^{N\left(n\alpha - \Omega\left(\frac{t^6}{n^3} \frac{\binom{n}{k}}{\binom{n}{k-t}}\right)\right)}.
  \end{align*}
  Estimating, we get that $\frac{\binom{n}{k}}{\binom{n}{k-t}}\geq 2^{-O(t)}$ and $\alpha \leq \frac{100 \log n}{(k/t)^{t}}\leq n^{-\Omega(t)}$,
  hence $n\alpha - \Omega\left(\frac{t^6}{n^3} \frac{\binom{n}{k}}{\binom{n}{k-t}}\right) \leq - 2^{-O(t)}$ and
  plugging this above yields that the left hand side of~\eqref{eq:6} is upper bounded by $2^{-2^{-O(t)}N}$.
  Thus, we conclude that
  \[
  \Prob{\mathcal{S}}{\exists \mathcal{T} \text{ with $\mu_{k-t}(\mathcal{T}) \leq 1/2$ such that }\nu(\partial^{t\cdot u} \mathcal{T})< \mu(\mathcal{T})}
  \leq \sum\limits_{N\geq 2^{0.8 n}}2^{-2^{-O(t)}N}
  \leq 2^{-2^{-O(t)}2^{0.8 n}}
  \leq 2^{-2^{n/2}}.
  \]

  \paragraph{The case that $\mu_{k-t}(\mathcal{T})> 1/2$.}
  Let $\mathcal{T}$ be such that $\mu_{k-t}(\mathcal{T})> 1/2$, and suppose that $\nu(\partial^{t\cdot u} \mathcal{T})< \mu(\mathcal{T})$.
  The analysis is similar to before, except that we look at $\mathcal{R} = \overline{\partial^{t\cdot u} \mathcal{T}}$ instead of $\mathcal{T}$.
  Thus, we get that $\nu(\mathcal{R}) > \mu(\overline{\mathcal{T}})$, and we argue that $\partial^{t\cdot d} \mathcal{R}\subseteq \overline{\mathcal{T}}$.
  Indeed, if $x\in \partial^{t\cdot d} \mathcal{R}$, then there is $y\in \mathcal{R}$ such that $x<y$, and as $y\in \mathcal{R}$ it follows that
  $y\notin \partial^{t\cdot u} \mathcal{T}$ so for all $x'<y$ of Hamming weight $(k-t)$ --- and in particular for $x' = x$ ---
  we have that $x'\not\in \mathcal{T}$, so $x\in \overline{\mathcal{T}}$.

  Thus, it follows that
  $\mu(\partial^{t\cdot d} \mathcal{R})<\nu(\mathcal{R})$ and now $\mu_{k}(\mathcal{R}) =
  \mu_k(\overline{\partial^{t\cdot u} \mathcal{T}}) = 1-\mu_k(\partial^{t\cdot u} \mathcal{T})\leq 1-\mu_{k-t}(\mathcal{T})\leq 1/2$,
  and the rest of the argument is analogous to the previous argument. Let $N = \card{\mathcal{R}}$.
  First, we have
  \begin{equation}\label{eq:7}
  \mu_{k-t}(\partial^{t\cdot d} \mathcal{R})\leq \frac{1}{t} \mu(\partial^{t\cdot d} \mathcal{R})\leq \frac{1}{t} \nu(R)\leq \frac{t+1}{t} \mu_{k}(\mathcal{R}).
  \end{equation}

  On the other hand, using Lemma~\ref{lemma:KK} we have
  $\mu_{k-t}(\partial^{t\cdot d} \mathcal{R})\geq \mu_{k}(\mathcal{R})^{\left(1-\frac{1}{n}\right)^t}$. As before, this implies
  $\mu_{k}(\mathcal{R})\geq e^{-O(n/t^2)}$ and so $N\geq e^{-O(n/t^2)} \binom{n}{k}\geq 2^{0.8 n}$. Also, it implies
  \begin{equation}\label{eq:8}
        \mu_{k-t}(\partial^{t\cdot d} \mathcal{R})
        \geq \mu_{k}(R) (1/2)^{\left(1-\frac{1}{n}\right)^t-1}\geq \mu_{k}(\mathcal{R}) 2^{t/n}.
  \end{equation}
  We now conclude from~\eqref{eq:8} that
  \begin{align*}
  0< \nu(\mathcal{R}) - \mu(\partial^{t\cdot d} \mathcal{R})
  &= t\mu_{k}(\mathcal{R}) + \mu_{k}(\mathcal{R}\cap \mathcal{S}) - (t+\rho)\mu_{k}(\partial^{t\cdot d} \mathcal{R})\\
  &\leq \mu_{k}(\mathcal{R}\cap \mathcal{S}) - (t+\rho-t2^{-t/n})\mu_{k-t}(\partial^{t\cdot d} \mathcal{R})\\
  &\leq \mu_{k}(\mathcal{R}\cap \mathcal{S}) - (\rho+t^2/2n)\mu_{k-t}(\partial^{t\cdot d} \mathcal{R}),
  \end{align*}
  so analogously to~\eqref{eq:4} we get that
  \begin{equation}\label{eq:9}
    \mu_{k}(\mathcal{R}\cap \mathcal{S})\geq (\rho+t^2/2n)\mu_{k-t}(\partial^{t\cdot d} \mathcal{R})\geq (\rho+t^2/2n)\mu_{k}(\mathcal{R}).
  \end{equation}

  As~\eqref{eq:7} holds, using Lemma~\ref{lem:sparse_approx2}, we get that there is $\mathcal{M}$ of size at most
  $\alpha N$ for $\alpha = \frac{100\log t}{\binom{k}{t}}$ satisfying the conclusion of the lemma for $\mathcal{B}_{\mathcal{M}}$ and $\mathcal{B}_{\mathcal{M}}'$ as therein. It follows that
  \begin{align*}
  \mu_{k}(\mathcal{B}_{\mathcal{M}}\cap \mathcal{S})
  \geq
  \mu_{k}(\mathcal{R}\cap \mathcal{S})
  -
  \mu_{k}(\mathcal{R} \Delta \mathcal{B}_{\mathcal{M}})
  &\geq \left(\rho + \frac{t^2}{2n} - \frac{18}{t}\right)\mu_{k}(\mathcal{R})\\
  &\geq \left(\rho + \frac{t^2}{2n} - \frac{36}{t}\right)\mu_{k}(\mathcal{B}_{\mathcal{M}})\\
  &\geq \left(\rho + \frac{t^2}{3n} \right)\mu_{k}(\mathcal{B}_{\mathcal{M}}),
  \end{align*}
  where we used the fact that $t\geq 10 n^{1/3}$. We may assume $\rho < 1- \frac{t^2}{3n}$, otherwise this is impossible.
  We denote $\rho' = \rho + \frac{t^2}{6n}$, so that now we are guaranteed that $t^2/6n\leq \rho'\leq 1- \frac{t^2}{6n}$.
  We also get that
  \[
  \mu_{k}(\mathcal{B}_{\mathcal{M}})\geq \frac{1}{2}\mu_{k}(\mathcal{R})\geq \frac{N}{2\binom{n}{k}}.
  \]
  Denote by $E_{\mathcal{M}}$ the event that
  $\mu_k(\mathcal{B}_{\mathcal{M}}\cap \mathcal{S})\geq \left(\rho + \frac{t^2}{3n} \right)\mu_k(\mathcal{B}_{\mathcal{M}})$.
  We now apply the Chernoff argument again; letting $\mathcal{S}'\subseteq \binom{n}{k}$ be chosen randomly by including each set
  with probability $\rho'' = \rho' + \frac{t^2}{12n}$, we get by Chernoff's bound that $\card{\mathcal{S}'}\geq s$ except with probability
  at most $1/2$ and so
  \[
  \Prob{\mathcal{S}}{\mathcal{E}_{\mathcal{M}}}
  \leq 2\Prob{\mathcal{S}'}{\mathcal{E}_{\mathcal{M}}}
  \leq 2e^{-\Omega\left(\frac{t^4}{n^2}\rho'\mu_k(\mathcal{B}_{\mathcal{M}})\binom{n}{k}\right)}
  \leq e^{-\Omega\left(\frac{t^6}{n^3}\frac{\binom{n}{k}N}{\binom{n}{k}}\right)}.
  \]
  Thus, by the union bound
  \[
  \Prob{\mathcal{S}}{\exists \mathcal{R}\text{ of size $N$ such that $\mu(\partial^{t\cdot d} \mathcal{R})<\nu(\mathcal{R})$}}
  \leq
  \sum\limits_{\card{\mathcal{M}}\leq \alpha N}\Prob{\mathcal{S}}{\mathcal{E}_{\mathcal{M}}}
  \leq 2^{n\alpha N - \Omega\left(\frac{t^6}{n^3}N\right)},
  \]
  and by a direct computation the last expression is at most
  $2^{-2^{-O(t)}N}$. Summing over $N\geq 2^{0.8 n}$ yields that
  \begin{align*}
  \Prob{\mathcal{S}}{\exists \mathcal{T}\text{ such that $\mu_{k-t}(\mathcal{T})\geq 1/2$, $\mu(\partial^{t\cdot u} \mathcal{T})>\nu(\mathcal{T})$}}
  \leq
  &\Prob{\mathcal{S}}{\exists \mathcal{R}\text{ such that $\mu_k(\mathcal{R})\leq 1/2$, $\mu(\partial^{t\cdot d} R)<\nu(\mathcal{R})$}}\\
  \leq
  &\sum\limits_{N\geq 2^{0.8 n}} 2^{-2^{-O(t)}N},
  \end{align*}
  which is at most $2^{-2^{n/2}}$ provided that $n_0$ is large enough.
\end{proof}

We will also need a version of Lemma~\ref{lem:check_hall} that works the other way around -- namely one that matches a slice and a random subset of
it with a slice above it, and we state it separately below.
\begin{lemma}\label{lem:check_hall2}
  For all $C>0$ there is $n_0\in\mathbb{N}$ such that the following holds.
  Let $n,k,s,t\in\mathbb{N}$ and assume that $\frac{n}{2}-C\sqrt{n\log n}\leq k\leq \frac{n}{2} + C\sqrt{n\log n}$
  and $10 n^{1/3}\leq t\leq C\sqrt{n\log n}$. Then for every $0\leq s\leq \binom{n}{k}$, setting $\rho = \frac{s}{\binom{n}{k}}$ we have
  \[
        \Prob{\substack{\mathcal{S}\subseteq \binom{[n]}{k}\\ \card{\mathcal{S}} = s}}{t\mu_k + \mu_\mathcal{S}\lesssim (t+\rho)\mu_{k+t}}
        \geq 1-2^{-\Omega(2^{n/2})}.
  \]
\end{lemma}
\begin{proof}
  Let $\mathcal{S}' = \sett{[n]\setminus A}{A\in \mathcal{S}}$ and note that it is a random subset of $\binom{n}{n-k}$ of size $s$, so
  applying Lemma~\ref{lem:check_hall} on $n-k$ instead of $k$ we get that with probability at least $1-2^{-\Omega(2^{n/2})}$
   there is a monotone fractional matching $w(x,y)$ from $(t+\rho)\mu_{n-k-t}$ to
  $t\mu_{n-k} + \mu_{\mathcal{S}'}$. Define $w'(A,B) = w(\overline{B},\overline{A})$, and note that it is a monotone fractional
  matching from $t\mu_{k} + \mu_\mathcal{S}$ to $(t+\rho)\mu_{k+t}$
\end{proof}

\subsection{Matching Union of Slices and a Random Subset to a Slice}
Next, we use Lemma~\ref{lem:check_hall} to show that given a union of consecutive slices and a random subset of
the topmost one, one can find a monotone fractional matching with each of the following: (1) a slice which is a bit above them, and (2) a slice which is a bit below them.
\begin{corollary}\label{corr:use_match}
  For all $C>0$ there is $n_0\in\mathbb{N}$, such that the following holds for all $n\geq n_0$.
  Let $n,t,k,s$ be as in Lemma~\ref{lem:check_hall}, let $\mathcal{S}$ be random subset of $\binom{[n]}{k}$ of size $s$
  and let $t\leq d\leq k/2$ be a parameter such that $\sum\limits_{i=k-d}^{k-1}\binom{n}{i}\geq t\binom{n}{k}$. Denote $\mathcal{T} = \mathcal{S}\cup \bigcup_{i=k-d}^{k-1}\binom{[n]}{i}$, and
  let $\nu_\mathcal{T}$ be the uniform distribution over $\mathcal{T}$. Then
  \[
  \Prob{\mathcal{S}}{\mu_{k-2d}\lesssim\nu_\mathcal{T}\lesssim \mu_{k+d}}\geq 1-2^{-\Omega(2^{n/2})}.
  \]
\end{corollary}
\begin{proof}
  We show that with probability $1-2^{-\Omega(2^{n/2})}$ we have that $\nu_\mathcal{T}\lesssim \mu_{k+d}$,
  and also that with probability $1-2^{-\Omega(2^{n/2})}$ we have that $\mu_{k-2d}\lesssim\nu_\mathcal{T}$.
  The statement then follows from the union bound.

  For the first statement, let $\rho = s/\binom{n}{k}$. By Lemma~\ref{lem:path_dom} we have that $\mu_{i}\lesssim \mu_{k}$ for $i\leq k$, so using Lemma~\ref{lem:basic_prop_matchings} we get that
  \[
  \nu_{\mathcal{T}}
  = \frac{\binom{n}{k}}{\card{\mathcal{T}}}\left(\rho \mu_S + \sum\limits_{i=k-d}^{k-1}\frac{\binom{n}{i}}{\binom{n}{k}}\mu_i\right)
  \lesssim
  \frac{\binom{n}{k}}{\card{\mathcal{T}}}\left(\rho \mu_S + \sum\limits_{i=k-d}^{k-1}\frac{\binom{n}{i}}{\binom{n}{k}}\mu_k\right)
  =\frac{\binom{n}{k}}{\card{\mathcal{T}}}\left(\rho \mu_S + \frac{\card{\mathcal{T}} - \rho\binom{n}{k}}{\binom{n}{k}}\mu_k\right).
  \]
  By Lemma~\ref{lem:check_hall2} we have that
  $t\mu_k+\rho\mu_{\mathcal{S}}\lesssim (t+\rho)\mu_{k+t}$ with probability $1-2^{-\Omega(2^{n/2})}$, in which case
  we get that
  \[
  \nu_{\mathcal{T}} \lesssim \frac{\binom{n}{k}}{\card{\mathcal{T}}}\left((t+\rho)\mu_{k+t} + \frac{\card{\mathcal{T}} - \rho\binom{n}{k}-t\binom{n}{k}}{\binom{n}{k}}\mu_k\right),
  \]
  where we used the fact that $\card{\mathcal{T}} - \rho\binom{n}{k} = \sum\limits_{i=k-d}^{k-1}\binom{n}{i}\geq t\binom{n}{k}$.
  Using $\mu_{k}\lesssim \mu_{k+t}$ and Lemma~\ref{lem:basic_prop_matchings} again and then simplifying, we conclude that
  $\nu_{\mathcal{T}}\lesssim \mu_{k+t}$.

  For the second statement, we note that by Lemmas~\ref{lem:path_dom},~\ref{lem:check_hall} we have that $(t+\rho)\mu_{k-t}\lesssim t\mu_k+\rho\mu_{\mathcal{S}}$
  with probability at least $1-2^{-\Omega(2^{n/2})}$.
  \begin{claim}\label{claim:bit_tricky}
    If $(t+\rho)\mu_{k-t}\lesssim t\mu_k+\rho\mu_{\mathcal{S}}$, then
    $(t+\rho)\mu_{k-t-d}\lesssim t\mu_{k-d}+\rho\mu_{\mathcal{S}}$.
  \end{claim}
  \begin{proof}
  Let $w(x,y)$ be a weight function showing that $(t+\rho)\mu_{k-t}\lesssim t\mu_k+\rho\mu_{\mathcal{S}}$.
  We consider the probability distribution $p(x,y) = \frac{1}{t+\rho} w(x,y)$, and define a probability distribution $p'$ over $(x',y')$ as follows:
  \begin{enumerate}
    \item Sample $(x,y)\sim p$ and independently a random permutation $\pi$ on $[n]$.
    \item Let $J\subseteq [n]$ be the set of first $d$ coordinates according to $\pi$ wherein $x_{j} = 1$.
    We define $x'_i = x_i$ on $i\not\in J$ and $x'_i = 0$ on $i\in J$.
    \item If $y\in \mathcal{S}$, with probability $1/2$ take $y' = y$.
    Otherwise, let $J'\subseteq [n]$ be the set of first $d$ coordinates according to $\pi$ wherein $y_j = 1$,
    and take $y'$ to be the vector where $y_i' = y_i$ on $i\not\in J'$ and $y_i' = 0$ on $i\in J'$.
  \end{enumerate}
  We argue that $w'(x,y) = (t+\rho)p'(x,y)$ shows that the fractional monotone matching as stated in the claim exists.
  For $y\in \mathcal{S}$ we have
  \[
  \sum\limits_{x'} w'(x',y) = \frac{1}{2}\sum\limits_x w(x,y) = \rho = (t\mu_{k-d}+\rho\mu_{\mathcal{S}})(y),
  \]
  and for $y'\not\in\mathcal{S}$ we have that $\sum\limits_{x'} p'(x',y')$ is the probability that we pick
  $y$ according to $\mu_k$, turn from $1$ to $0$ a random set of $d$ coordinates and reach $y'$, which is the $\mu_{k-d}(y')$. Thus,
  $\sum\limits_{x'} w'(x',y')=(t\mu_{k-d}+\rho\mu_{\mathcal{S}})(y')$.

  For $x'$, $\sum\limits_{y'} p'(x',y')$ is the probability we take $x\sim \mu_{k-t}$, turn from $1$ to $0$ a random set of $d$ coordinates
  and reach $x'$, which is equal to $\mu_{k-t-d}(x')$, hence $\sum\limits_{y'} w'(x',y')=(t+\rho)\mu_{k-t-d}(x')$.
  \end{proof}

  Using Claim~\ref{claim:bit_tricky} we get by Lemmas~\ref{lem:basic_prop_matchings},~\ref{lem:path_dom}
  \begin{align*}
  \mu_{k-2d}
  &\lesssim
  \frac{\binom{n}{k}}{\card{\mathcal{T}}}\left(t\mu_{k-d}+\rho\mu_{\mathcal{S}}+\frac{\sum\limits_{i=k-d}^{k-1}\binom{n}{i} - t\binom{n}{k}}{\binom{n}{k}}\mu_{k-2d}\right)\\
  &\lesssim
  \frac{\binom{n}{k}}{\card{\mathcal{T}}}\left(t\mu_{k-d}+\rho\mu_{\mathcal{S}}+\frac{\sum\limits_{i=k-d}^{k-1}\binom{n}{i} - t\binom{n}{k}}{\binom{n}{k}}\mu_{k-d}\right)\\
  &=
  \frac{\binom{n}{k}}{\card{\mathcal{T}}}\left(\rho\mu_{\mathcal{S}}+\frac{\sum\limits_{i=k-d}^{k-1}\binom{n}{i}}{\binom{n}{k}}\mu_{k-d}\right)\\
  &\lesssim
  \frac{\binom{n}{k}}{\card{\mathcal{T}}}\left(\rho\mu_{\mathcal{S}}+\frac{\sum\limits_{i=k-d}^{k-1}\binom{n}{i}\mu_i}{\binom{n}{k}}\right)
  ,
  \end{align*}
  which is equal to $\nu_{\mathcal{T}}$.
\end{proof}

\subsection{Proof of Theorem~\ref{thm:main_matchings_monotone}}
In this section, we prove Theorem~\ref{thm:main_matchings_monotone} which by Lemma~\ref{lem:finish} implies Theorems~\ref{thm:2},~\ref{thm:directed_isoperimetric}.
Theorem~\ref{thm:main_matchings_monotone} is a direct consequence of the following more precise statement:
\begin{thm}
  There exists $C>0$ such that for all $m\in\mathbb{N}$ and all $n\geq C\cdot m^6$, there are
  $m$ sets $P_1,\ldots,P_m\subseteq\{0,1\}^n$ satisfying the following properties:
  \begin{enumerate}
    \item $P_i\cap P_j = \emptyset$ for all $i\neq j$.
    \item $\card{P_i} = \left\lfloor \frac{2^n}{m}\right\rfloor$ for all $i$.
    \item For each $i$, $x\in P_i$ and $y\in P_{i+1}$ we have that $\card{x}\leq \card{y}$.
    For all $x\in P_1\cup\ldots\cup P_m$ and $y\not\in P_1\cup\ldots\cup P_m$
    we have that $\card{x}\leq \card{y}$.
    \item For each $i$ there is a monotone matching from $P_i$ to $P_{i+1}$.
  \end{enumerate}
    In particular, the following function $\phi$ is monotone and admits an $m2^{-n}$-almost perfect matching:
    $\phi(x) = i-1$ if $x\in P_i$, and otherwise $\phi(x) = m-1$.
\end{thm}
\begin{proof}
  We present a randomized construction and show that it works with probability $1-o(1)$.
  Consider a random $\pi\colon \{0,1\}^n\to[2^n]$ such that $\pi(x) \leq \pi(y)$ whenever $\card{x}\leq \card{y}$;
  in other words, we first think of an ordering of $\{0,1\}^n$ as $x_0,\ldots,x_{2^n-1}$, where we first enumerate
  according to Hamming weight and within each layer we order randomly. Thus, we may take $\pi(x_i) = i$.

  Define the $P_1,\ldots,P_m$ as
  \[
  P_i = \pi^{-1}\left(\sett{\left\lfloor \frac{2^n}{m}\right\rfloor(i-1) + j}{j=1,\ldots,\left\lfloor \frac{2^n}{m}\right\rfloor}\right),
  \]
  so that the first three items holds trivially. In the rest of the proof, we argue that the fourth item holds with probability $1-o(1)$.
  Denote by $\ell_i$ and $u_i$ the smallest and largest Hamming weight of vectors from $P_i$, and by $\nu_i$ the uniform distribution over
  $P_i$. It suffices to prove that with probability $1-2^{-\Omega(2^{n/2})}$ we have that $\nu_{i}\lesssim \nu_{i+1}$ for all $i$
  Indeed, then we get by Lemma~\ref{lem:frac_to_matching} that there is a monotone matching from $P_i$ and $P_{i+1}$, and the fourth item follows.

  We now show that for each $i$, $\nu_i\lesssim \nu_{i+1}$ with probability $1-2^{-\Omega(2^{n/2})}$, and then the claim follows by the union bound.
  We intend to use Corollary~\ref{corr:use_match}
  to show that and therefore we break each one of $\nu_i$, $\nu_{i+1}$ into lower and upper part.
  Let $m_i$ be the median Hamming weight of $\nu_i$, namely
  such that $\nu_{i}(\sett{x}{\card{x}< m_i})<0.5$ but $\nu_{i}(\sett{x}{\card{x}\leq m_i})\geq 0.5$. Define $p_i = \nu_{i}(\sett{x}{\card{x}< m_i})$,
  $q_i = \nu_{i}(\sett{x}{\card{x}> m_i})$, and let
  \[
  \nu_{i}^{-}(x) = \nu_i(x) 1_{\card{x}< m_i} + \left(\frac{1}{2} - p_i\right)\mu_{m_i}(x),
  \qquad
  \nu_{i}^{+}(x) = \nu_i(x) 1_{\card{x}> m_i} + \left(\frac{1}{2} - q_i\right)\mu_{m_i}(x).
  \]

  Let $t = \lceil 10 n^{1/3} + 1\rceil$. Our goal is to show that with probability $1-2^{-\Omega(2^{n/2})}$, for all $i$ we have
  \begin{equation}\label{eq:10}
  \frac{1}{2}\mu_{\ell_i - 2t}\lesssim \nu_{i}^{-},
  \qquad
  \nu_{i}^{-}\lesssim \frac{1}{2}\mu_{m_i},
  \qquad
  \nu_{i}^{+}\lesssim \frac{1}{2}\mu_{u_i+2t},
  \qquad
  \frac{1}{2}\mu_{m_i}\lesssim \nu_{i}^{+},
  \end{equation}
  in which case we get, using Lemma~\ref{lem:basic_prop_matchings}, that
  \[
  \nu_i =
  \nu_{i}^{-}
  +
  \nu_{i}^{+}
  \lesssim
  \frac{1}{2}\mu_{m_i}
  +
  \frac{1}{2}\mu_{u_i+2t}
  \lesssim
  \frac{1}{2}\mu_{\ell_{i+1}-2t}
  +
  \frac{1}{2}\mu_{m_{i+1}}
  \lesssim
  \nu_{i+1}^{-}
  +
  \nu_{i+1}^{+}
  =
  \nu_{i+1}.
  \]
  Here, we also the facts that $m_i\leq \ell_{i+1}-2t$, $u_i+2t\leq m_{i+1}$ and Lemma~\ref{lem:path_dom}. This follows since the
  probability mass of each layer in the hypercube is at most $O(1/\sqrt{n})$, hence each $P_i$ must intersect at least
  $\Omega(\sqrt{n}/m)$ distinct layers and so $u_i - m_i\geq \Omega(\sqrt{n}/m) \geq \Omega(C^{1/6}n^{1/3}) > 100t$,
  and in the same way $m_i - \ell_i > 100t$ and $\ell_{i+1} - m_i\geq 100t$. We also note that all of the $u_i,m_i$ and $\ell_i$'s are all in the range
  $[\frac{n}{2} - \sqrt{100n\log n},\frac{n}{2} + \sqrt{100n\log n}]$ since the total probability mass outside this range is at most
  $e^{-\frac{1}{2} 100 \log n}\leq n^{-10} < 1/m$.

  We finish by arguing that~\eqref{eq:10} holds for each $i$ with probability $1-2^{-\Omega(2^{n/2})}$, and for that we apply Corollary~\ref{corr:use_match}.
  Set $d=2t$; we argue that for each $\ell_1\leq k\leq u_m$ it holds that
  \[
  \sum\limits_{i=k-d}^{k-1}\binom{n}{i}\geq  t\binom{n}{k}.
  \]
  Indeed, this follows since the ratio between any two consecutive binomial coefficients $\binom{n}{i}$ for $i=k-d,\ldots,k$
  is $1+O(\sqrt{\log n/n})$, so the ratio between any two (not necessarily consecutive) binomial coefficients in that range is at most $\left(1+O(\sqrt{\log n/n})\right)^{d} = 1+o(1)$,
  so $\sum\limits_{i=k-d}^{k-1}\binom{n}{i}\geq (1+o(1)) d\binom{n}{k} = (2+o(1))t\binom{n}{k}>t\binom{n}{k}$. Thus, the conditions of
  Corollary~\ref{corr:use_match} hold, and applying it for various $k$'s we get that~\eqref{eq:10} holds with probability $1-2^{-\Omega(2^{n/2})}$.
  Below, we explain in details how to deduce that $\frac{1}{2}\mu_{\ell_i - 2t}\lesssim \nu_{i}^{-}$, and the other arguments are similar.

  We view $\nu_i(x) 1_{\card{x}< m_i}$ as a uniform weight function over
  the part of $P_i$ of Hamming weight less than $m_i$, which is a union of slices and a random subset of the appropriate
  size of the slice $\ell_i$. Thus by Corollary~\ref{corr:use_match} we get that $p_i\mu_{\ell_{i}-2t}\lesssim \nu_i(x) 1_{\card{x}< m_i}$
  with probability $1-2^{-\Omega(2^{n/2})}$, and as $\mu_{\ell_{i}-2t}\lesssim \mu_{m_i}$ we get from Lemma~\ref{lem:basic_prop_matchings} that
  $\frac{1}{2}\mu_{\ell_{i}-2t}\lesssim \nu_i^{-}$.

\end{proof}

\bibliographystyle{plain}
	\bibliography{refs}
\end{document}